\documentclass[11pt]{article}
\usepackage[english]{babel}
\usepackage{verbatim}
\usepackage{amsmath, amsthm, amssymb, graphics, graphicx}
\usepackage{enumitem}

\newcommand{\mathsym}[1]{{}}

\newtheorem{theorem}{Theorem}[section]
\newtheorem{lemma}[theorem]{Lemma}

\theoremstyle{remark}
\newtheorem{claim}{Claim}

\title{\bf\Large{{Connected greedy colouring in claw-free graphs}}} 

\author{Ngoc Khang Le and Nicolas Trotignon\thanks{CNRS, LIP, ENS de Lyon. Partially supported by the LABEX MILYON
    (ANR-10-LABX-0070) of Universit\'e de Lyon, within the program
    ‘‘Investissements d'Avenir’’ (ANR-11-IDEX-0007) operated by the
    French National Research Agency (ANR).  Also Universit\'e Lyon~1,
    Universit\'e  de Lyon. E-mail:
  nicolas.trotignon@ens-lyon.fr}}

\newcommand{\epc}{This proves Claim~\theclaim.}

\begin{document}
\maketitle

\begin{abstract}
  An ordering of the vertices of a graph is \emph{connected} if every vertex
  (but the first) has a neighbor among its predecessors.  The greedy
  colouring algorithm of a graph with a connected order consists in
  taking the vertices in order, and assigning to each vertex the
  smallest available colour. A graph is \emph{good} if the greedy algorithm on every connected order gives every connected induced subgraph of it an optimal colouring. We give the characterization of good claw-free graphs in terms of minimal forbidden induced subgraphs.
\end{abstract}

\section{Introduction}

A \emph{$k$-colouring} for a graph $G$ is any function $\pi$ from
$V(G)$ to $\{1, \dots, k\}$ such that for any edge $uv \in E(G)$,
$\pi(u) \neq \pi(v)$.  The smallest integer $k$ such that $G$ admits a
$k$-colouring is called the \emph{chromatic number} of $G$ and is
denoted by $\chi(G)$.  A $\chi(G)$-colouring of $G$ is called an
\emph{optimal colouring of~$G$}.  Computing the chromatic number is
known to be difficult.

 Let $G$ be a graph and ${\cal O}= [v_1, \dots, v_n]$ be a linear
 ordering of its vertices.  The \emph{greedy colouring algorithm}
 (greedy algorithm for short) applied to $(G, {\cal O})$ consists in taking
 the vertices in the order $\cal O$, and giving to each vertex a colour
 equal to the smallest positive integer not used by its neighbours
 already coloured.  This obviously produces a colouring.

 For every graph, there exists an order $\cal O$ for the vertices such
 that the greedy algorithm produces an optimal colouring.  To see
 this, consider an optimal colouring $\pi$, and consider the following
 ordering: first take vertices with colour 1, then vertices with colour
 2, and so on.  But this method has no practical interest to compute
 optimal colourings, since to find the ordering, an optimal colouring has
 to be known.  

 It is also well known that for some graphs, there exist orderings that
 produce colourings very far from the optimal, for instance consider
 two disjoint sets on $n$ vertices, say $A=\{a_1, \dots, a_n\}$ and
 $B= \{b_1, \dots, b_n\}$.  Add all possible edges between $A$ and
 $B$, except edges $a_ib_i$, $i\in \{1, \dots, n\}$.  This produces a
 bipartite graph $G$.  However, the greedy algorithm applied to the
 order $[a_1, b_1, a_2, b_2, \dots, a_n, b_n]$ produces a colouring
 with $n$ colours.

One might wonder for which graphs the greedy algorithm always gives an optimal solution no matter what order is given. The operation
 \emph{Disjoint-Union} consists in building a new graph by taking the
 union of two vertex-disjoint graphs.  The operation
 \emph{Complete-Join} consists in building a new graph by taking the
 union of two vertex-disjoint graphs $G_1$ and $G_2$, and by adding
 all possible edges between $V(G_1)$ and $V(G_2)$.
Let $P_k$ denote the path on $k$ vertices.  When $H$ and $G$ are
 graphs, we say that $G$ is \emph{$H$-free} if $G$ does not contain an
 induced subgraph isomorphic to $H$.  A \emph{cograph} is a $P_4$-free
 graph.  Seinsche~\cite{S74} proved that cographs are exactly
 the graphs that can be produced by starting with graphs on one vertex
 and by repeatedly apply the operations Disjoint-Union and
 Complete-Join to previously constructed graphs.
The graphs such that the greedy algorithm on every order gives every induced subgraph of them an optimal colouring are fully characterized. 

\begin{theorem}[see \cite{W90,CS79}]
  \label{th:Ccograph}
  For every graph $G$, the following properties are equivalent.
  \begin{itemize}
  \item $G$ is a cograph.
  \item For every induced subgraph $H$ of $G$ and every linear order
    $\cal O$ of $V(H)$, the greedy colouring algorithms applied to
    $(H, {\cal O})$ produces an optimal colouring of $H$.
  \end{itemize}
\end{theorem}

There are many ways to order the vertices of a graph with the hope to obtain a better colouring. In this paper, we focus on \emph{connected} orders. An order ${\cal O} = [v_1, \dots, v_n]$ for a graph $G$ is
\emph{connected} if for every $2 \leq i \leq n$, there exists $j<i$
such that $v_jv_i \in E(G)$.  A connected order exists if and only if
$G$ is connected, and is efficiently produced by search algorithms
such as BFS, DFS (or more simply by the algorithm \emph{generic
  search}).  
We say that a graph $G$ is \emph{good} if for every connected induced subgraph $H$ of $G$ and every connected order $\cal O$ of $H$, the greedy algorithm produces an optimal colouring of $H$. Also, a connected order $\cal O$ of a graph $G$ is \emph{good} if it produces an optimal colouring of $G$. A graph or a connected order is \emph{bad} if it is not good. A graph is \emph{minimally bad} if it is bad and all other connected induced subgraphs of it are good. Connected orders are better than general orders for
colouring bipartite graphs. 

\begin{theorem}[see \cite{BCD14}]
  \label{th:bip}
  Every bipartite graph is good. 
\end{theorem}

However, unlike general orders, it is not true that for every graph, there exists a connected
order that provides an optimal colouring,
see~\cite{BT94} for example.  A similar
claw-free example is given here: 

\begin{figure}[h]
\centering
\includegraphics[width=6cm]{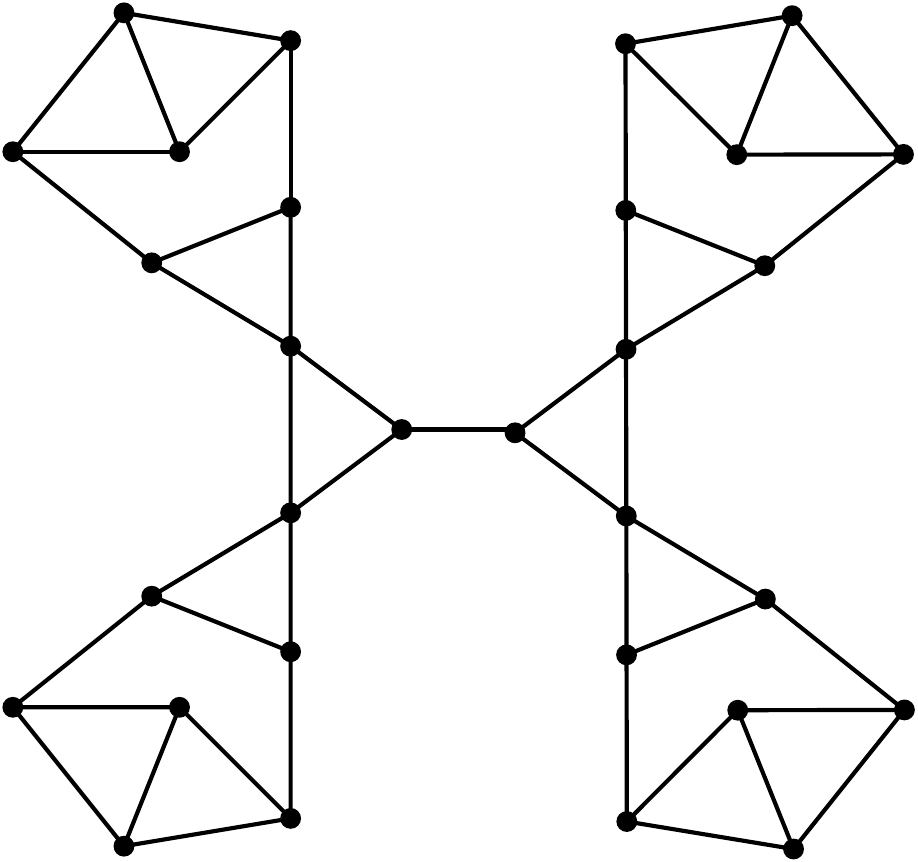}
\caption{A claw-free graph where every connected order is bad.}
\label{F:3}
\end{figure}

The connected greedy colouring has recently been studied. In \cite{BCD14}, they define $\Gamma_c(G)$ as the maximum number $k$ such that there exists a connected order producing a $k$-colouring of $G$. They also proved that checking if $\Gamma_c(G)\geq k$ is NP-hard if $k$ is a part of the input. In \cite{BFKS15}, they show that this problem remains NP-hard even when $k=7$. A graph $G$ is good in our definition if for every connected induced subgraph $H$ of $G$, $\Gamma_c(H)=\chi(H)$. Note that their results imply also that checking if there exists a bad connected order for a graph is NP-hard, but do not imply NP-hardness on recognizing good graphs (since a class of good graphs is hereditary by our definition). The complexity of recognizing good graphs remains open. In \cite{BT94}, they gave several examples of small graphs that are not friendly with connected orders. They also proved that gem (see Figure \ref{F:2}) is the unique smallest bad graph. In \cite{HW89}, they defined a more restricted good graph with respect to connected orders and gave the complete characterization of this class. Therefore, their class is also good by our definition. 

However, the list of excluded induced subgraphs for the class of good graphs is still unknown.  Equivently, no description of minimally bad graphs is known. Our goal is to prove an analogue of Theorem~\ref{th:Ccograph} for connected orders. 
If we restrict our attention to claw-free graphs, we are able to give this description (where the \emph{claw} is the graph on $\{a, b, c, d\}$ with edges $ab$, $ac$ and $ad$). This is our main result that we now state precisely.  The rest of the paper is devoted
to its proof.

\subsection*{The main result}

Let $G(V,E)$ be a graph. For $v\in V(G)$, let $N(v)$ denote the set of vertices in $G$ that are adjacent to $v$. For $S\subseteq V(G)$, we denote by $G[S]$ the subgraph of $G$ induced by $S$. A subset $K\subseteq V(G)$ is a  \emph{clique} in $G$ if all vertices in $K$ are pairwise adjacent. Let $A,B\subseteq V(G)$, we say that $A$ is \emph{complete} to $B$ if for every $x\in A$ and $y\in B$, $xy\in E(G)$. If $A=\{x\}$, we also say that $x$ is complete to $B$ instead of saying $\{x\}$ is complete to $B$.

A \emph{cycle} in $G$ is a sequence of distinct
vertices $v_1\dots v_k$ such that $v_iv_{i+1}\in E(G)$ for
$i\in \{1,\ldots,k\}$ (the index is taken modulo $k$).  The edges
$v_iv_{i+1}$ are the edges of the cycles, the other edges between the
vertices of the cycle are called its \emph{chords}.  The length of a
cycle is the number of its edges (here $k$).  A \emph{hole} is a cycle
of length at least~4 that has no chord. A \emph{path} in $G$ is a
sequence $P = v_1\ldots v_k$ of distinct vertices of $G$ such that
$v_iv_j\in E(G)$ if and only if $|i-j|=1$ (paths are often refered to as induced
path or chordless paths).  Vertices $v_1$ and $v_k$ are the
\textit{ends} of $P$ and the rest of the vertices are
\emph{internal}. The length of a path is the number of its edges.  A
hole (cycle, path) is \textit{even} or \textit{odd} according to the
parity of its length. When $P = v_1 \dots v_k$ is a path, and
$1\leq i \leq j \leq k$, the path $v_i \dots v_j$ is called the
subpath of $P$ from $v_i$ to $v_j$ and denoted by $v_iPv_j$. A path in
a graph $G$ is \textit{flat} if all its internal vertices are of
degree $2$ (in $G$). A \emph{triangle} is a graph on three vertices and they are all adjacent.

A graph $H$ is a \textit{prism} if:
\begin{itemize}
	\item $V(H)=V(P_1)\cup V(P_2)\cup V(P_3)$.
	\item For $i\in \{1,2,3\}$, $P_i$ is a path of length $\geq 1$ with two ends $a_i$ and $b_i$.
	\item $P_1$, $P_2$, $P_3$ are vertex-disjoint.
	\item $\{a_1,a_2,a_3\}$ and $\{b_1,b_2,b_3\}$ are two triangles.
	\item These are the only edges in $H$.
\end{itemize}

A prism is \emph{short} if one of its three paths is of length
$1$. A prism is \textit{parity} if its three paths have the same
parity and is \textit{imparity} otherwise.  Note that a prism contains
an odd hole if and only if it is imparity.  A parity prism is
\textit{even} (\textit{odd}) if the lengths of its three paths are
even (odd). 

We also need several particular graphs, defined in Figure~\ref{F:2}. 

\begin{figure}[h]
rs\centering
\includegraphics[width=11cm]{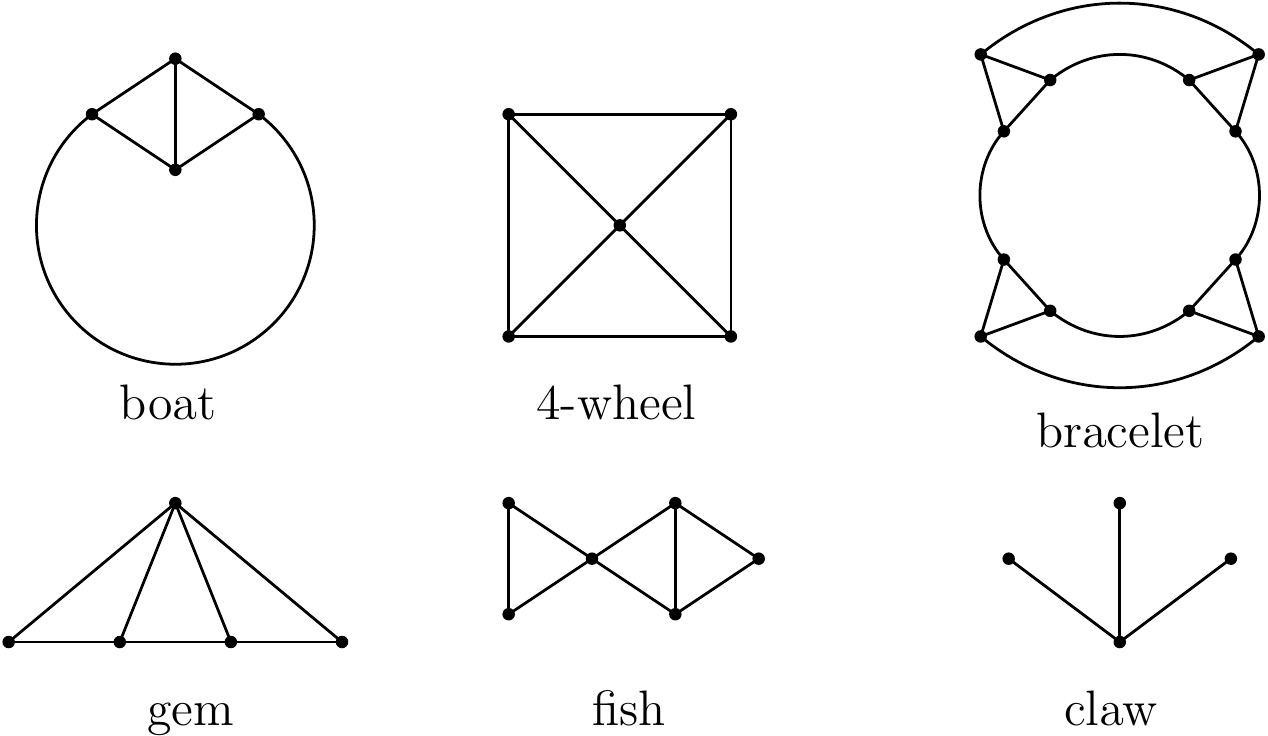}
\caption{Some graphs}
\label{F:2}
\end{figure}

We call \emph{obstructions} the graphs represented in Figure~\ref{F:1}
with the following additional specifications:

\begin{itemize}
\item The orientation represented for each graph has no special
  meaning. It is an indication of how a bad connected order can be
  found for it. The orientation does not fully specify this order.  The
  arrow should be seen from a small to a big vertex with respect to this order. The chromatic
  number of each graph is $3$ and the last vertex in every bad order
  receives colour $4$.
\item All the straight lines are edges, all the curved lines are paths
  of length $\geq 1$.
\item The hole in $F_1$ is odd.
\item The only path in $F_2$ is of length $\geq 1$. The orientation of the only unoriented edge depends on the parity of this path. $F_2$ is a gem when the
  length of this path is $1$.
\item The only path in $F_3$ is of length $\geq 1$.
\item The hole in $F_5$ is even.
\item All paths in $F_7$, $F_8$, $F_9$, $F_{10}$ are of length $\geq 2$.
\item $F_7$ is an imparity prism. The lower path is of different
  parity from the other two paths.
\item The prism in $F_8$ is an even prism. The upper path of the prism
  contains two flat paths: the first one is odd, the second is even.
\item The prisms in $F_9$ and $F_{10}$ are odd prisms.
\item The upper path of the prism in $F_9$ contains two odd flat paths.
\item The upper and lower paths of the prism in $F_{10}$ contain four even flat paths.
\item The length of the only long cycle in $F_{11}$ is odd $\geq
  3$. If its length is $3$, then $F_{11}$ is a fish.
\item The length of two flat paths in $F_{12}$ is odd $\geq 3$.
\end{itemize}

\begin{figure}[h]
\centering
\includegraphics[width=11cm]{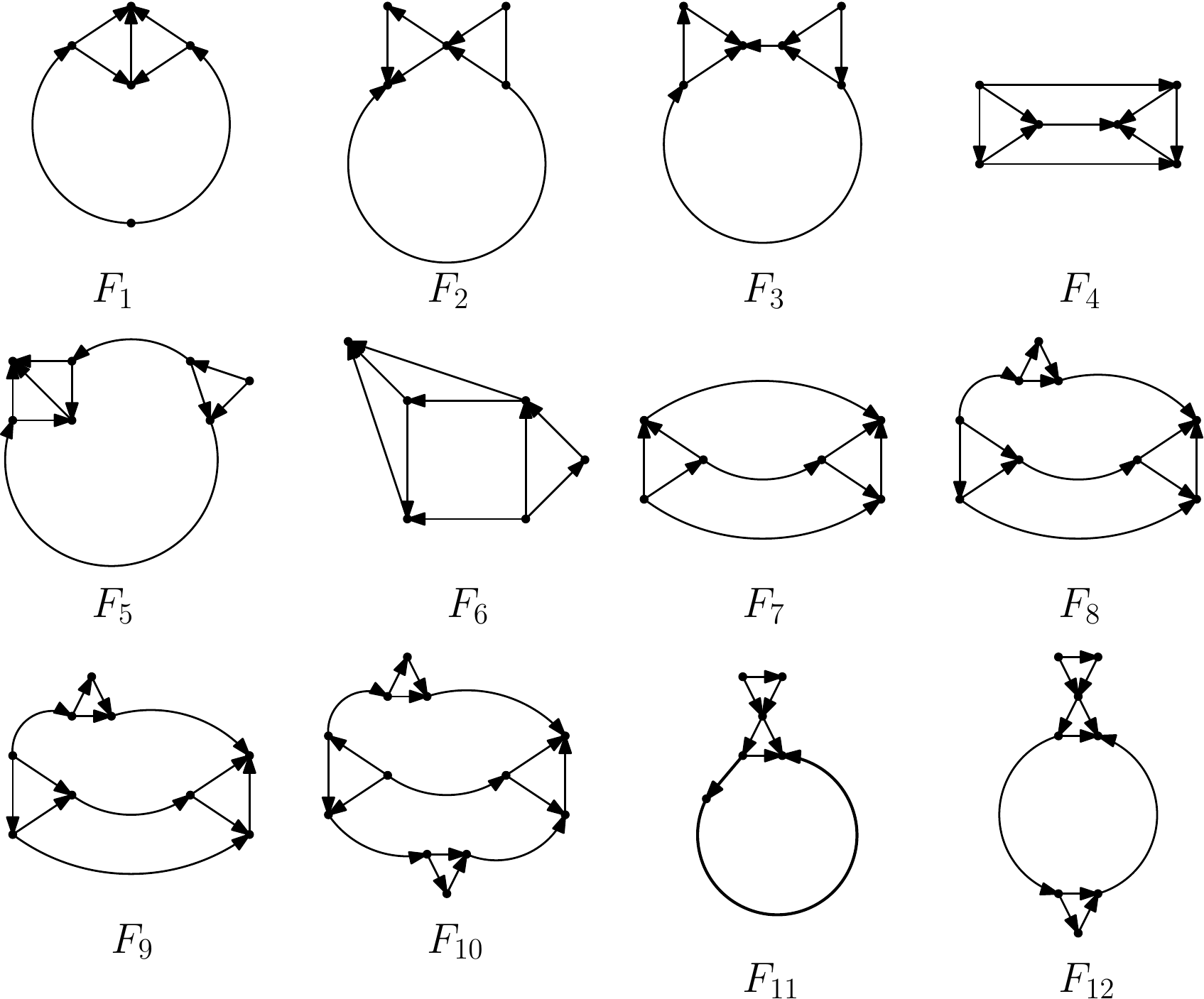}
\caption{List of obstructions}
\label{F:1}
\end{figure}

Our main result is the following. 

\begin{theorem} \label{T1} Let $G$ be a claw-free graph. Then $G$ is
  good if and only if $G$ does not contain any obstruction as an induced subgraph.
  Equivalently, a claw-free graph is minimally bad if and only if it is an
  obstruction.
\end{theorem}

\section{Some properties of minimally bad graphs}

For any graph $G$, any order $\cal O$ of its vertices and any vertex
$v$, let $\pi_{G, \cal O}(v)$ be the colour that vertex $v$ receives
when applying the greedy colouring algorithm to $G$ with order
$\cal O$.  We also write $\pi(v)$ or $\pi_{\cal O}(v)$ when the
context is clear.

Let $G$ be a graph with an ordering ${\cal O}= [v_1, \ldots, v_n]$ of
its vertices.  For vertices $u, v$ of $G$, we use the notations
$u<_{\cal O}v$, $u>_{\cal O}v$, $u\leq_{\cal O}v$, $u\geq_{\cal O}v$
with the obvious meaning.  When clear from the context, we omit the
subscript~${\cal O}$.  When $v$ is a vertex of $G$, we denote by
$G_{\leq v}$ the subgraph of $G$ induced by
$\{u\in V(G) \text{ such that } u \leq v\}$.  Similarly, we use the
notations $G_{< v}$, $G_{\geq v}$ and $G_{> v}$.

When $X\subseteq V(G)$, we use the notation ${\cal O}[X]$ to denote
the order induced by $\cal O$ on $X$, and ${\cal O} \setminus X$ to
denote the order induced by $\cal O$ on $V(G) \setminus X$. We write
${\cal O}\setminus v$ instead of ${\cal O}\setminus \{v\}$.  We denote
by $\max(X)$ (resp.\ $\min(X)$) the maximum (resp.\ minimum) element
in $X$.

Let $G$ be a graph and ${\cal O}= [v_1, \dots, v_n]$ be a linear ordering
of its vertices.  The \emph{greedy colouring algorithm starting with
  colour 2} applied to $(G, {\cal O})$ consists in giving $v_1$ colour
2, and then taking the vertices from $v_2$ on in the order ${\cal O}$,
and to give to each vertex a colour equal to the smallest positive
integer not used by its neighbours already coloured.  

\begin{lemma} 
  \label{Lm:start}
  When applied to a good graph, the greedy colouring algorithm starting with
  colour 2 produces an optimal colouring.
\end{lemma}

\begin{proof}
  The colouring produced by this algorithm is the same as the colouring
  produced by the connected order obtained from ${\cal O}$ by swapping
  the first two vertices. Hence it is optimal.
\end{proof}

For the rest of this section, $G$ is a minimally bad graph with a bad
order ${\cal O}= [v_1,\ldots, v_n]$. Note that for any set
$S\subsetneq V(G)$, if ${\cal O}[S]$ is a connected order then it
produces an optimal colouring for $G[S]$.

\begin{lemma} 
  \label{Lm2} 
  For every $x\in V(G)\setminus \{v_n\}$, $\pi(x)\leq \chi(G)$ and
  $\pi(v_n)=\chi(G)+1$.
\end{lemma}

\begin{proof}
  Follows directly from the fact that $\cal O$ is a bad order and that
  $G$ is a minimally bad graph.
\end{proof}

\begin{lemma}
  \label{l:remove}
  If $x\in V(G) \setminus \{v_n\}$ and ${\cal O}\setminus x$ is
  connected, then for some vertex $y\neq x$ in $G$,
  $\pi_{G\setminus x, {\cal O}\setminus x}(y) \neq \pi_{G, {\cal
      O}}(y)$.
\end{lemma}

\begin{proof}
  The conclusion is true for $y=v_n$. Because by the minimality of
  $G$, we have $\pi_{G\setminus x, {\cal O}\setminus x}(v_n) \leq
  \chi(G\setminus x) \leq \chi(G)$ and by Lemma~\ref{Lm2}, $\pi(v_n) =
  \chi(G) +1$. 
\end{proof}

\begin{lemma} 
  \label{C2} 
  $\pi(v_n)\geq 4$.
\end{lemma}

\begin{proof}
  Otherwise, $\pi(v_n)\leq 3$, so by Lemma~\ref{Lm2}, $\chi(G)\leq 2$,
  so $G$ is bipartite, a contradiction to Theorem~\ref{th:bip}.
\end{proof}

\begin{lemma} 
  \label{Lm1} 
  For every vertex $v\in V(G)$, $G_{\leq v}$, $G_{\geq v}$, $G_{<v}$
  and $G_{>v}$ are connected. In particular, $G$ is connected. 
\end{lemma}

\begin{proof}
  For $G_{\leq v}$, it comes from the definition of connected orders.

  Suppose $C_1,\ldots,C_k$ ($k\geq 2$) are the connected components of
  $G_{\geq v}$.  For $i=1, \dots k$, set
  $G_i=G[\{u \in V(G) \text{ such that } u<v\} \cup C_i]$ and let
  ${\cal O}_i$ be the order $\cal O$ restricted to $V(G_i)$. For every
  $i\in\{1,\ldots,k\}$ and for every vertex $u\in C_i$, we have
  $\pi_{G, \cal O}(u)=\pi_{G_i, {\cal O}_i}(u)$ because there are no
  edges in $G$ between $C_i$ and $C_j$ for $i\neq j$.  But since $G$
  is minimally bad, $V(G_i)\subsetneq V(G)$ and ${\cal O}_i$ is a
  connected order, $\pi_{{\cal O}_i}$ is an optimal colouring for
  $G_i$. So, for every vertex $u$ in $G$,
  $\pi(u) \leq \chi(G_i) \leq \chi(G)$, so $\pi$ is an optimal
  colouring, a contradiction.

  The proof is the same for $G_{<v}$ and $G_{>v}$ (note that we view
  the empty graph as a connected graph). 
\end{proof}

A \emph{cutset} in a graph $G$ is a set $S\subseteq V(G)$ such that
$G\setminus S$ is disconnected. A cutset $S$ is a \emph{clique cutset} if $S$ is a clique.

\begin{lemma}
  \label{C1} 
  If $S$ is a cutset of $G$, then for every component $C$ of
  $G\setminus S$ except at most one, $\max(C)<\max(S)$.  Furthermore,
  if $C$ is the unique component such that $\max(C)>\max(S)$, then
  $v_n\in C$.
\end{lemma}

\begin{proof}
  For the first claim, if $\max(C)>\max(S)$ for more than one component
  $C$, then $G_{>\max(S)}$ is disconnected, a contradiction to
  Lemma~\ref{Lm1}.  The second claim follows trivially.
\end{proof}

\begin{lemma}
  \label{ColClique}
  Suppose $S$ is a clique cutset of $G$ and $C$ is a component of
  $G\setminus S$ such that $\max(S) < \min(C) = v$.  If 
   $v$ is complete to $S$, then there exists $u\in S\cup \{v\}$ such that
  $\pi(u) > |S| + 1$.
\end{lemma}

\begin{proof}
  Otherwise, since $S\cup \{v\}$ is a clique, the colours
  $1, \dots, |S|+1$ are exactly the colours used in $S\cup \{v\}$. Now
  build an order ${\cal O}'$ of $G[S\cup C]$ by first reordering the
  vertices from $S\cup \{v\}$ by increasing order of their colours, and
  then  taking the rest of $S\cup C$ as it is ordered by $\cal O$.
  This new order is connected (as $\cal O$) and therefore provides an optimal
  colouring of $G[S\cup C]$.  It also gives the same colouring as
  $\cal O$ for $G[S\cup C]$.  Since by Lemma~\ref{C1} $v_n\in C$, it
  follows that $\pi(v_n) \leq \chi(G[S\cup C]) \leq \chi(G)$, a
  contradiction to Lemma~\ref{Lm2}.
\end{proof}

\begin{lemma} 
  \label{Lm5}
  For $v\in V(G)$, let $S$ be a cutset of $G_{\leq v}$.  If there
  exists a connected component $C$ of $G_{\leq v}\setminus S$ such
  that $\min(C) < \min(S)$ then $v_1\in C$.
\end{lemma}

\begin{proof}
  If $v_1\notin C$, then $G_{<\min(S)}$ is not connected: $v_1$ and
  $\min(C)$ are in different components, a contradiction to
  Lemma~\ref{Lm1}.
\end{proof}

It is sometimes convenient to view $G$ and $\cal O$ as an oriented
graph $D_G$, obtained from $G$ by orienting from $u$ to $v$ every edge
$uv$ such that $u<v$. We therefore use the notion of
\emph{in-neighbor}, \emph{outneighbor}, \textit{source} and
\textit{sink} in $G$ (a source in $G$ is a vertex with no in-neighbor
in $D_G$ and a sink in $G$ is a vertex with no outneighbor in $D_G$).

\begin{lemma}
  \label{CUnique} 
  $G$ has a unique source that is $v_1$ and a unique sink that is $v_n$. 
\end{lemma}

\begin{proof}
  Obviously, $v_1$ is a source and $v_n$ is a sink.  If $G$ has two
  sources $u<v$, then $G_{\leq v}$ is disconnected ($u$ and $v$ are in
  two distinct components), a contradiction to Lemma~\ref{Lm1}.  If
  $G$ has two sinks $u<v$, then $G_{\geq u}$ is disconnected ($u$ and
  $v$ are in two distinct components), a contradiction to
  Lemma~\ref{Lm1}.
\end{proof}

\begin{lemma} 
  \label{Lm3} 
  Let $v$ be a vertex of degree $2$ in $G$ and let $b<a$ be its
  neighbors. One and exactly one of the following outcome occurs:
  \begin{itemize}
  \item $v=v_1$ is the source of $G$ and $v_2=b$;
  \item $b < v < a$.
  \end{itemize}
  Moreover, $\pi(v)\in\{1,2\}$.
\end{lemma}

\begin{proof}
  If $b < a <v$, then $v$ is a sink of $G$ and $v=v_n$ by Lemma
  \ref{CUnique}. Since $v$ has degree~2, $\pi(v)\leq 3$, a
  contradiction to Lemma~\ref{C2}.
 
  If $v<b< a$ then $v$ is a source of $G$ and $v=v_1$ by Lemma
  \ref{CUnique}. Hence, $\pi(v)=1$. Also, $v_2 = b$ because $\cal O$
  is connected.
  
  Otherwise, $b < v < a$.  So, $v$ has degree~1 in $G_{\leq v}$ and
  $\pi(v) \in \{1, 2\}$. 
\end{proof}

\begin{lemma} 
  \label{C3} 
  In colouring $\pi$, the colours of the internal vertices of any flat
  path in $G$ alternates between $1$ and $2$.
\end{lemma}

\begin{proof}
  Clear by Lemma~\ref{Lm3}.
\end{proof}

\begin{lemma}
  \label{l:endFP}
  If $P$ is a flat path of $G$, then $\max(V(P))$ is an end of $P$. 
\end{lemma}

\begin{proof}
  If $P$ has length at most~1, the conclusion is trivial. Otherwise,
  the ends of $P$ form a cutset of $G$ (note that $G=P$ is impossible
  since a path is a good graph by Theorem~\ref{th:bip}).  If
  $\max(V(P))$ is not an end of $P$, then by Lemma~\ref{C1}, $v_n$ is
  an internal vertex of $P$.  So, by Lemma~\ref{C3},
  $\pi(v_n)\in \{1, 2\}$, a contradiction to Lemma~\ref{C2}.
\end{proof}

A path $P= p_1 \dots p_k$ in $G$ is \emph{well ordered} if
$p_1 < p_2 < \dots < p_k$ or $p_k < \dots < p_2 < p_1$. A flat path in $G$ is \emph{maximal} if its two end are not of degree $2$ in $G$.

\begin{lemma}
  \label{l:Pwo}
  If $P = a \dots b$ is a flat path in $G$ then either it is well
  ordered, or the source $v_1$ is an internal vertex of $P$ and $aPv_1$, $v_1Pb$
  are both well ordered.  In particular, there exists at most one
  maximal flat path in $G$ that is not well ordered.
\end{lemma}

\begin{proof}
  This follows from Lemmas~\ref{Lm3} and the definition of
  connected orders.
\end{proof}

\begin{lemma} 
  \label{Lm6} 
  Let $k\geq 2$ and $S= \{s_1, \dots, s_k\}$ be a set of vertices in
  $G$ such that $s_1< \dots < s_k$ and
  $s_k$ is complete to $\{s_1, \dots, s_{k-1}\}$. Let
  $a_1, \ldots,a_k\in G\setminus S$ be $k$ distinct vertices of degree
  $2$ in $G$ and such that $N_S(a_i)=s_i$. Suppose that for
  $i = 1 \dots k-1$, $N(s_i)\setminus \{a_i,s_k\}\subseteq N(s_k)$.  If
  $a_k < s_k$ then:
  \begin{enumerate}[label=(\arabic*)]
  \item \label{Lm6:c1} For every $v\in N(s_k)\setminus \{a_k\}$ such
    that $v < s_k$, $\pi(v)\neq \pi(a_k)$.
  \item \label{Lm6:c2} $\pi(a_1)= \ldots=\pi(a_k)=1$ or
    $\pi(a_1)= \ldots = \pi(a_k)=2$. In particular, $\{a_1,
    \dots, a_k\}$ is a stable set of $G$. 
\end{enumerate}
\end{lemma}

\begin{proof}
  To prove~\ref{Lm6:c1}, suppose that there exists a vertex
  $v\in N(s_k)\setminus \{a_k\}$ such that $v < s_k$ and
  $\pi(v)= \pi(a_k)$.  Let $b\neq s_k$ be the second neighbor of
  $a_k$.  Since $a_k < s_k$, by Lemma~\ref{Lm3}, $a_k$ is the source
  of $G$, or $b < a_k < s_k$. In either case, we can see that
  ${\cal O}\setminus a_k$ is a connected order for $G\setminus a_k$,
  because $s_{k-1} < s_k$ (and $k\geq 2$).
 
  If $a_k$ is not the source of $G$, order ${\cal O}\setminus a_k$
  gives an optimal colouring $\pi'$ of $G\setminus a_k$ because $G$ is
  minimally bad.  Morevover, for every vertex $u\neq a_k$ in $G$, we
  have $\pi'(u) = \pi(u)$. For $u=b$ this is because $b<a_k$, for the
  other $u<s_k$ this is because $a_k$ brings no constraint to $u$ and for
  $u=s_k$, this is because the only constraint brought by $a_k$ is
  also brought by $v$ (because $\pi(v) = \pi(a_k)$).  So, $\pi$ is an
  optimal colouring of $G$, a contradiction.

  If $a_k$ is the source of $G$, then $b=v_2$ because $s_{k-1} < s_k$.
  Hence, $\pi(b)=2$. We consider the greedy algorithm starting with color $2$ applied to $(G\setminus a_k,{\cal O}\setminus a_k)$. This is a
  connected order, and it therefore provides an optimal colouring
  $\pi'$ of $G\setminus a_k$ by Lemma \ref{Lm:start}.  Again, for every vertex of $G$,
  $u\neq a_k$, we have $\pi'(u) = \pi(u)$, because the only constraint
  brought by $a_k$ is given to $s_k$, and $v$ gives the same
  constraint.  So, $\pi$ is an optimal colouring of $G$, a
  contradiction.
  
  Let us now prove~\ref{Lm6:c2}.  By Lemma~\ref{Lm3}, we know that
  for $i=1, \dots ,k$, $\pi(a_i)=1$ or $\pi(a_i)=2$. If $\pi(a_k)=1$,
  then suppose that for some $i<k$, $\pi(a_i)=2$. No neighbor of $s_i$
  smaller than $s_i$ has colour 1: for $a_i$ by assumption, and all
  others are in $N(s_k) \setminus \{a_k\}$, so we know this
  by~\ref{Lm6:c1}.  Hence, $\pi(s_i)=1$, contradicting~\ref{Lm6:c1}.
  If $\pi(a_k)=2$, the proof is similar.
 \end{proof}

\begin{lemma}
  \label{Lm7}
  Suppose that $G$ is claw-free.  Let $s_1, s_2$ be two vertices in
  $G$ such that $s_1< s_2$ and $s_1s_2\in E(G)$. Let $a_1, a_2$ be
  distinct vertices of degree $2$ in $G$, such that
  $a_1s_1, a_2s_2 \in E(G)$ and $a_2<s_2$.  Suppose that
  $N(s_1) \setminus \{a_1, s_2\} = N(s_2) \setminus \{a_2, s_1\} = K$,
  where $K$ is a non-empty clique. Suppose that $\{s_1, s_2\}$ is a cutset
  in $G$ and $C_1$, $C_2$ are two connected components of
  $G\setminus \{s_1, s_2\}$ such that $a_1, a_2 \in C_1$ and $K\subseteq C_2$.

  So, $\pi(a_1)=\pi(a_2)=2$, $s_1<a_1$ and there exist vertices
  $v\in K$, $p, q\in C_2\setminus K$ such that $vpq$ is a triangle,
  $v<s_1$ and $v<s_2$.
\end{lemma}

\begin{proof}
  We first prove that $v=\min(K)<s_2$.  Otherwise, $s_2 <v$. Also,
  $v=\min(C_2)$ because $\cal O$ is connected.  In $G_{\leq s_2}$, $s_1$
  and $s_2$ both have degree at most 2, so
  $\pi(s_1), \pi(s_2) \in \{1, 2, 3\}$. In $G_{\leq v}$, $v$ has
  degree 2, so $\pi(v) \in \{1, 2, 3\}$.  Hence, the clique cutset
  $S= \{s_1, s_2\}$, $C_2$ and $v$ contradict Lemma~\ref{ColClique}.
  This proves our claim.

 By Lemma~\ref{Lm6}, we consider two cases.

\smallskip
  \noindent{\bf Case 1:} $\pi(a_1)=\pi(a_2)=1$.

  By Lemma~\ref{Lm6}, $\pi(v)\neq 1$.  So, there exists
  $x$ adjacent to $v$ with $x<v$ and $\pi(x)=1$.
  Note that $x\notin K$ because $x<v$ and $x\notin\{s_1, s_2\}$
  because $\pi(x) = 1$.  If $s_1<v$, then $G_{<v}$ is
  disconnected ($x$ and $s_1$ are in different components). Therefore,
  $v<s_1$.  We then have $s_1<a_1$ for otherwise, $G_{<s_1}$ is
  disconnected ($a_1$ and $v$ are in different components).  So,
  $\pi(s_1)=1$ (since no vertex smaller than $s_1$ in $K$ has
  colour $1$ by Lemma \ref{Lm6}).  This is a contradiction because
  $\pi(a_1)=1$.  

\smallskip
  \noindent{\bf Case 2:} $\pi(a_1)=\pi(a_2)=2$.

  First, $\pi(s_1)=1$ since if $\pi(s_1)\geq 2$, there exists a vertex
  $u\in K$ such that $u<s_1$ and $\pi(u)=1$.  So, $s_1 < a_1$ for
  otherwise $G_{<s_1}$ is disconnected.  Since by Lemma~\ref{Lm6} no
  vertex smaller than $s_1$ in $K$ receives colour $2$,  $s_1$
  receives colour~2, a contradiction. 

  By Lemma~\ref{Lm6}, $\pi(v)\neq 2$ and because of $s_1$,
  $\pi(v)\neq 1$.  So, $\pi(v)\geq 3$.  Hence, some in-neighbor $q$ of
  $v$ ($q\notin K$) satisfies $\pi(q)=2$. If $s_1<v$, then $G_{<v}$ is
  disconnected ($s_1$ and $q$ are in different components).  So,
  $v<s_1$. Therefore, $v$ must have an in-neighbor $p\notin K$, with
  $\pi(p)=1$.  Now, $pq\in E(G)$ since $G$ is claw-free.  Finally,
  $s_1< a_1$, for otherwise $G_{< s_1}$ is disconnected ($v$ and
  $a_1$ are in different components).
\end{proof}

\section{Forbidden structures of minimally bad graphs}

\label{S:3}

Throughout this section, let $G$ be a minimally bad claw-free graph that is not an obstruction. 

A graph $H$ is a \textit{cap} in $G$ if:
\begin{itemize}
	\item $V(H)=K\cup V(P)$.
	\item $K$ is a clique disjoint from $P$, $K=L\cup R\cup C$
          such that $L$, $R$, $C$ are non-empty.
	\item $P$ is a flat path in $G$ of odd length $\geq 1$ with two ends $a$, $b$.
	\item $a$ is complete to $L$, $b$ is complete to $R$.
	\item These are the only edges in $H$.
	\item No vertex in $L \cup R \cup V(P)$ has a neighbor in
          $G\setminus H$.
\end{itemize} 

\begin{lemma} \label{Lm:cap}
$G$ does not contain a cap.
\end{lemma}

\begin{proof}
  \setcounter{claim}{0}
  Suppose $G$ contains a cap $H$ and $K$, $L$, $R$, $C$, $P$, $a$, $b$
  are defined as in the definition of a cap. Let $a'$ and $b'$ be the
  vertices adjacent to $a$ and $b$ in $P$, respectively.  By
  Lemma~\ref{l:endFP}, we may assume up to symmetry that
  $b=\max(V(P))$.

  \begin{claim} 
    \label{cap:c2}
    For every vertex $v\in K$, if $v<b$, then $\pi(v)\neq \pi(b)$.
  \end{claim}

    Otherwise, there exists $v<b$ such that $\pi(v) = \pi(b)$. Note in
    particular that by Lemma~\ref{Lm2}, $b\neq v_n$.  Also, the
    existence of $v$ implies that ${\cal O}\setminus b$ is a connected order for $G \setminus b$.  We then see that for every vertex
    $y\neq b$,
    $\pi_{G\setminus b, {\cal O}\setminus b}(y) = \pi_{G, {\cal
        O}}(y)$, a contradiction to Lemma~\ref{l:remove}.
    \epc

  \begin{claim} 
    \label{cap:c3} If $\pi(a)=1$ or $\pi(a)=2$, then every
    vertex $v\in K$ satisfies $v>b$.
  \end{claim}

    Since $\pi(a)=1$ or $\pi(a)=2$, we have $\pi(b') = \pi(a)$ by
    Lemma~\ref{C3} and the parity of $P$. Suppose that there exists a
    vertex $v\in K$ with $v<b$.  Then, there exists a vertex $u\in L$
    such that $u<b$, for otherwise $G_{<b}$ is disconnected ($v$ and
    $a$ are in different components). Since by Claim~\ref{cap:c2},
    $\pi(u)\neq \pi(b)$ and $u$ has no neighbor with colour $\pi(b)$,
    we have $\pi(u) < \pi(b)$. But then, when the greedy algorithms
    visits $b$, colour $\pi(u)$ is available for $b$ (because
    $\pi(u) \neq \pi(a) = \pi(b')$ and $u$ is complete to $R$), a
    contradiction.
\epc

  By Claim~\ref{cap:c3}, if $\pi(a)=1$ or $\pi(a)=2$, then
  $\pi(b)=3-\pi(a)$ by the parity of $P$.  Since $\cal O$ is
  connected, in fact for every vertex $v$ in $V(G) \setminus V(P)$, $v > b$.
  So, when the greedy algorithm visits $G\setminus P$, the first
  vertex receives colour 1 or 2, and it gives exactly the same colours as the greedy algorithm starting with color 1 or 2 applied to ($G\setminus P$, ${\cal O}\setminus P$)
  . Hence, by Lemma~\ref{Lm:start}, we see that $\cal O$
  is a good order for $G$, a contradiction. Hence, $\pi(a)\geq 3$.

  This implies that $a$ has degree at least $2$ in $G_{\leq a}$, so
  $a$ has an in-neighbor $u$ in $L$.  So, if $a'<a$, then $G_{< a}$
  is diconnected ($a'$ and $u$ are in different components). Hence,
  $a<a'$.  So there exists a vertex in $L$ with colour $1$ (to ensure that
  $a$ has colour at least~3). Since $a < a'$, we know by
  Lemma~\ref{l:Pwo} that $P$ is well ordered. We therefore have
  $\pi(b')=\pi(a)\geq 3$ (if $b'=a$) or $\pi(b')=2$ by
  the parity of $P$. This implies $\pi(b)=1$, contradicting
  Claim~\ref{cap:c2}.
\end{proof}

A graph $H$ is an \textit{even birdcage} in $G$ if:
\begin{itemize}
	\item $V(H)=\cup_{i=1}^{k} V(P_i)  \cup C_a\cup C_b$ for some $k\geq 3$.
	\item $\forall i\in\{1,\ldots,k\}$, $P_i$ is a flat path in $G$ of even length $\geq 2$ with two ends $a_i$, $b_i$ (all $a_i$'s and $b_i$'s are distinct).
	\item $S_a=\{a_1,\ldots,a_k\}$ and $S_b=\{b_1,\ldots,b_k\}$ are two cliques.
	\item $K_a=C_a\cup S_a$ and $K_b=C_b\cup S_b$ are two cliques
          ($C_a$ and $C_b$ may be empty).
	\item If $C_a\neq \emptyset$, then $S_a$ is a clique cutset of $G$.
	\item If $C_b\neq \emptyset$, then $S_b$ is a clique cutset of $G$.
	\item These are the only edges in $H$.
	\item No vertex in $\cup_{i=1}^{k} V(P_i)$ has a neighbor in $G\setminus H$.
\end{itemize}

\begin{lemma} \label{Lm:even-birdcage}
$G$ does not contain an even birdcage.
\end{lemma}

\begin{proof}
  \setcounter{claim}{0}

  Suppose $G$ contains an even birdcage $H$, with the notation as in
  the definition of an even birdcage. Up to symmetry, we suppose
  $b_k = \max(S_a\cup S_b)$. 

  If $C_b\neq \emptyset$, let $C_1$ be the connected component of
  $G\setminus S_b$ that contains $S_a$ and let $C_2$ be the connected
  component of $G\setminus S_b$ that contains $C_b$.  Since $G$ is connected
  by Lemma~\ref{Lm1}, $C_1\neq C_2$. 

  If $C_b=\emptyset$, then $G\setminus S_b$ is connected (because $G$ is
  connected). We then set $C_1 = V(G) \setminus S_b$ and $C_2 = \emptyset$.

  \begin{claim}
    $\max(C_1)< b_k$.
  \end{claim}

    Set $v=\max(C_1)$.  If $v\in V(P_i)$ for some $i\in \{1, \dots,
    k\}$, then $v< b_k$ follows from Lemma~\ref{l:endFP}.  Otherwise,
    $G_{>\max(S_a)}$ is disconnected ($v$ and $b_k$ are in different
    components), a contradiction to Lemma~\ref{Lm1}.
\epc

For all $i\in\{1,\ldots,k\}$, let $a_i'$, $b_i'$ be the vertices in
$P_i$ adjacent to $a_i$, $b_i$ respectively. By Lemma
\ref{Lm6} applied to $S_b$, we consider the following two cases:

\smallskip
  \noindent{\bf Case 1:}
	For all $i\in\{1,\ldots,k\}$, $\pi(b_i')=1$.  

        Then, by the parity of $P$ and Lemma~\ref{C3},  we also have $\pi(a_i')=1$.

	\begin{claim} \label{Cl:2}
	$\forall i\in\{1,\ldots,k\}$, $b_i'<b_i$.
	\end{claim}

        Suppose for some $i\in\{1,\ldots,k\}$, $b'_i>b_i$. Then, when
        the greedy algorithm visits $b_i$, there must be an
        in-neighbor of $b_i$ with colour~1 (because $\pi(b_i) \neq
        1$). This vertex is an in-neighbor of  $b_k$ with colour~1, a
        contradiction to Lemma~\ref{Lm6}\ref{Lm6:c1}.
        \epc
        
	\begin{claim}
          If $C_b\neq\emptyset$, then $b_k  < \min(C_2)$.
	\end{claim}

        Set $v=\min(C_2)$. We know that $v_1\in C_1$ by Claim \ref{Cl:2} and Lemma \ref{Lm5} (applied to $S_b$). So, $v\in C_b$ since $\cal O$ is a
        connected order.  Also, since no vertex in $S_b$ has colour 1,
        $v$ receives colour 1.  Hence, by Lemma~\ref{Lm6}\ref{Lm6:c1},
        $b_k<v$. \epc

        \begin{claim}
          $C_a$ contains a vertex $x$ of colour $1$ (in particular,
          $C_a\neq\emptyset$).
        \end{claim}

        By Lemma~\ref{l:Pwo}, we may assume up to symmetry that $P_1$
        is well ordered.  Since $\pi(a_1)\neq 1$ and since colour 1
        does not appear in $S_a$, $a_1$ must have an in-neighbor in
        $C_a$ with colour~1. \epc

        Now if $C_b\neq \emptyset$, then $v=\min(C_b)$ receives colour
        1 and $b_1, \dots, b_k$ are coloured with colours
        $2, \dots, k+1$.  This contradicts Lemma~\ref{ColClique}.  If
        $C_b=\emptyset$ then $b_k=v_n$ and $\pi(v_n)=k+1\leq \chi(G)$
        (since $S_a\cup \{x\}$ is a clique of size $(k+1)$ in $G$),
        contradicting Lemma \ref{Lm2}.

        \smallskip
        \noindent{\bf Case 2:}
        $\pi(b_i')=2$ $\forall i\in\{1,\ldots,k\}$

        Then, by the parity of $P$ and Lemma~\ref{C3}, we also have
        $\pi(a_i')=2$. By Lemma~\ref{l:Pwo}, up to symmetry, we may
        assume that $P_1$ is well ordered. Suppose that
        $a_1>a'_1\geq b'_1 > b_1$ (the case $a_1<a'_1\leq b'_1 < b_1$ is similar).  Since $\pi(b'_1) = 2$, we have
        $\pi(b_1)=1$.  Again, by Lemma~\ref{l:Pwo} and up to symmetry,
        we may assume that $P_2$ is well ordered.  Since
        $\pi(b_2) \neq 1$, we must have $b_2>b'_2\geq a'_2 > a_2$ and
        $\pi(a_2) = 1$.  

        Now if $P_3$ is also well ordered, we must have $\pi(a_3) = 1$
        or $\pi(b_3) = 1$, a contradiction.  It follows by
        Lemma~\ref{l:Pwo} that the source of $G$ is an internal vertex
        of $P_3$.  By the parity of $P$, it follows that $\min(a_3,
        b_3)$ receives colour 1, a contradiction. 
\end{proof}

A graph $H$ is an \textit{odd birdcage} in $G$ if:
\begin{itemize}
\item
  $V(H)=C_a\cup C_b \cup (\cup_{i=1}^{k} V(P_i)) \cup
  (\cup_{i=1}^{m} K_i)$ for some $k\geq 3$, $m\geq 0$.
\item $P_1$ is a path of odd length $\geq 3$ with two ends $a_1$, $b_1$.
\item $\forall i\in\{2,\ldots,k\}$, $P_i$ is a flat path of odd length $\geq 3$ with two ends $a_i$, $b_i$. 
\item All $a_i$'s and $b_i$'s are distinct.
\item $S_a=\{a_1,\ldots,a_k\}$ and $S_b=\{b_1,\ldots,b_k\}$ are two cliques.
\item $K_a=C_a\cup S_a$ and $K_b=C_b\cup S_b$ are two cliques ($C_a$ and $C_b$ might be empty).
\item If $C_a\neq \emptyset$, then $S_a$ is a clique cutset of $G$.
\item If $C_b\neq \emptyset$, then $S_b$ is a clique cutset of $G$.
\item For $i\in\{1,\ldots,m\}$, $c_i$'s, $d_i$'s are  vertices of $P_1$ such that:
  \begin{itemize}
  \item They appear in $P_1$ in the following order: $a_1$, $c_1$, $d_1$,\ldots, $c_m$, $d_m$, $b_1$.
  \item $\forall i\in\{1,\ldots,m\}$, $c_id_i$ is an edge.
  \item $a_1P_1c_1$ and $d_mP_1b_1$ are flat paths of even length $\geq 2$.
  \item $\forall i\in\{1,\ldots,m-1\}$, $d_iP_1c_{i+1}$ is a flat path of odd length $\geq 3$.
  \item $\forall i\in\{1,\ldots,m\}$, $K_i$ is a non-empty clique complete to $\{c_i,d_i\}$.
  \item $\forall i\in\{1,\ldots,m\}$, $\{c_i,d_i\}$ is a cutset of $G$.
  \end{itemize}
\item These are the only edges in $H$.
\item No vertex in $\cup_{i=1}^k V(P_i)$ has a neighbor in $G\setminus H$.
\end{itemize}

\begin{lemma} \label{Lm:odd-birdcage}
  $G$ does not contain an odd birdcage.
\end{lemma}

\begin{proof}
  \setcounter{claim}{0}
  Suppose $G$ contains an odd birdcage $H$ as in the definition of an
  odd birdcage. For $i\in\{1,\ldots,k\}$, let $a_i'$, $b_i'$ be the
  neighbors of $a_i$, $b_i$ in $P_i$ respectively. For
  $j\in\{1,\ldots,m\}$, let $c_j'$, $d_j'$ be the neighbors of $c_j$,
  $d_j$ in $P_1\setminus\{c_j, d_j\}$ respectively. Let
  $x = \max(\cup_{i=1}^m V(P_i))$. By Lemma~\ref{l:endFP}, we may
  assume that $x=d_i$ for some $i\in \{1, \dots, m\}$ or $x=b_j$ for
  some $j\in \{1, \dots, k\}$.

\smallskip
  \noindent{\bf Case 1:} $x=d_i$ for some $i\in \{1, \dots, m\}$.

  By Lemma~\ref{Lm7} applied to the cutset $\{c_i,d_i\}$, there exist
  some vertices $u\in K_i$, $p,q\notin K_i$ such that $u,p,q$ is a
  triangle and $\pi(c_i')=\pi(d_i')=2$, $c_i<c_i'$, $u<c_i$ and
  $u<d_i$. Then $m=i=1$ for otherwise $G$ contains $F_{12}$.  Let $C$
  be the component of $G\setminus\{c_1, d_1\}$ that contains $K_1$.  By
  Lemma~\ref{Lm5}, the source of $G$ is in $C$. Hence, by
  Lemma~\ref{l:Pwo}, all flat paths in $H$ are well ordered. In particular,
  $c_1'P_1a_1$ is a directed odd path from $c_1'$ to $a_1$. Since
  $\cal O$ is connected, $c_1'P_1a_1$ contains the first vertices of
  $\cup_{i=1}^m V(P_i)$, so that $\pi(a_1)=1$.  Also, $b_1P_1d_1'$ is a
  directed odd path from $b_1$ to $d_1'$, so $\pi(b_1)=1$ because
  $\pi(d'_1)=2$.

  Let $b_l= \min(b_1, \dots, b_k)$.  Note that if
  $C_b \neq \emptyset$, $b_l < \min(C_b)$ for otherwise $G_{<b_l}$ is
  disconnected ($c_1$ and $\min(C_b)$ are in different components).
  So, by the connectivity of $\cal O$, $b_l$ has an in-neighbor in
  $P_l$. Since $b_1P_1d_1'$ is a directed odd path from $b_1$ to
  $d_1'$, $b_l\neq b_1$.  Also, $P_l$ is well ordered, and since
  $\pi(a_l)\neq 1$ (because of $a_1$), we have $\pi(a'_l)=1$.  By the
  parity of $P_l$, it follows that $\pi(b_l) = 1$, a contradiction
  since $\pi(b_1)=1$.  

\smallskip
  \noindent{\bf Case 2:} $x=b_j$ for some $j$.

  By Lemma \ref{Lm6}, $\pi(b_i')=1$ $\forall i\in\{1,\ldots,k\}$ or
  $\pi(b_i')=2$ $\forall i\in\{1,\ldots,k\}$. Suppose first that
  $\pi(b_i')=1$ $\forall i\in\{1,\ldots,k\}$.  

  Then for every $i\in\{2,\ldots,k\}$, $\pi(a_i')=2$ by the parity of
  the flat paths. Also, for every $i\in\{2,\ldots,k\}$,
  $b_i'<b_i$ since otherwise $b_i$ has colour $\geq 2$ and
  there must exist in $K_b$ some in-neighbor of $b_i$
  having colour $1$, contradicting Lemma~\ref{Lm6}.  By
  Lemma~\ref{l:Pwo}, we may assume that $P_3$ is well ordered.  Since
  $\pi(a'_3) = 2$, we have $\pi(a_3)=1$. So, $\pi(a_2) \neq 1$, and
  since $\pi(a'_2)=2$, we have $a'_2<a_2$.  So, $P_2$ is not
  well ordered. By Lemma~\ref{l:Pwo}, the source of $G$, $v_1$, is an
  internal vertex of $P_2$ and $v_1P_2a_2$ and $v_1P_2b_2$ are both
  well ordered. If $a_2<b_2$, then $\pi(a_2) = 1$, a contradiction.  Hence,
  $b_2<a_2$ and $\pi(b_2)=2$. The vertex $v$ that comes just after
  $b_2$ in $\cal O$ cannot be $a_2$, because then, again, we would
  have $\pi(a_2)=1$. Hence, $v$ is in $K_b$ and receives
  colour~1, a contradiction to Lemma~\ref{Lm6}.

  Suppose now that $\pi(b_i')=2$ $\forall i\in\{1,\ldots,k\}$. Then
  for every $i\in\{2,\ldots,k\}$, $\pi(a_i')=1$ by the parity of the
  flat paths.

  \begin{claim} 
    \label{odd-bridcage:c2} If $C_b\neq \emptyset$, then for every
    vertex $v\in C_b$, $v>b_j$.
  \end{claim}
        
	Set $v = \min (C_b)$ and suppose $v < b_j$. 

        If $\pi(v)=1$ then $v<b_i$ for every $i\in\{1,\ldots,k\}$ for otherwise, $\min(S_b)$ would
        receive colour~1, a contradiction. Hence, by Lemma~\ref{Lm5},
        the source of $G$ is in the component of $G\setminus S_b$
        containing $v$.  Also, by Lemma~\ref{Lm6}, no in-neighbor of
        $b_j$ in $K_b$ has colour 2.  So, $\min(S_b)$ has no coloured
        neighbor in $\{b'_1, \dots, b'_k\}$ when the greedy algorithm
        visits it, so it receives colour $2$, a contradiction.

        If $\pi(v)\geq 3$, then there exists some vertex $q\notin C_b$
        adjacent to $v$ having colour $2$. If for some $i$, $b_i<v$,
        then $G_{<v}$ is not connected ($q$ and $b_i$ are in different
        components), a contradiction. Then $v<\min(S_b)$, and because
        $\pi(v)\neq 1$ there exists a vertex $p\notin C_b$
        adjacent to $v$ with colour $1$. We have $pq\in E(G)$ because
        $G$ is claw-free. Therefore $G$ contains $F_{12}$, a
        contradiction. \epc

	\begin{claim} \label{odd-bridcage:c3}
	$\pi(d_m)=1$.
	\end{claim}	

	Assume that $\pi(d_m)\neq 1$. We have $\pi(d_m')=2$ by the
        parity of flat path $b_1'P_1d_m'$. If $d_m<d_m'$ then
        $\pi(d_m)=1$, a contradiction. Therefore, $d_m'<d_m$.  If
        $\pi(b_1)>1$, then $b_1'<b_1$ and by Lemma~\ref{l:Pwo}, 
        $v_1\in b_1'P_1d_m'$.  So, if $b_1<d_m$,
        $\pi(b_1)=1$, a contradiction, and otherwise $\pi(d_m)=1$, a
        contradiction again. Hence, $\pi(b_1)=1$. Therefore for every
        $i\neq 1$, $b_i'<b_i$ for otherwise, $b'_i$ receives
        colour~1. Now, consider the cutset
        $S = \{d_m,b_2,\ldots,b_k\}$. Let $C$ be the component of
        $G\setminus S$ that contains $c_m$ and $D$ the component that
        contains $b_1$.  Note that $d_m$ must have an in-neighbor in
        $C$, for otherwise it would receive colour 1. So, we have
        $\min (C) < \min (S)$.  Also, $\min(D) < \min (S)$ because
        every vertex in $S$ has an in-neighbor in $D$ (clear for
        $d_m$, for the other ones, it follows from the fact that they
        have an in-neighbor with colour 1, namely $b_1$).  This
        contradicts Lemma~\ref{Lm5}. \epc

	\begin{claim} \label{odd-bridcage:c4}
	$\pi(c_m')=1$.
	\end{claim}	
	Suppose $\pi(c_m')=2$. Then $c_m'<c_m$ and $d_m<c_m$, so there
        exists an outneighbor $v$ of $c_m$ in $K_m$ (otherwise
        $c_m=v_n$, contradicting the maximality of $b_j$). Therefore
        $G_{>c_m}$ is disconnected ($v$ and $b_j$ are in different
        components), contradicting Lemma~\ref{Lm1}. \epc

	Similarly, we can prove $\pi(d_i)=\pi(c_i')=1$ for all
        $i\in\{1,\ldots,m\}$, and therefore $\pi(a_1')=1$. 
        
        We must have a vertex of colour $1$ in $C_a$, otherwise
        $a_i'<a_i$ for every $i\in\{1,\ldots,k\}$.  Also, for some
        $1\leq i' < i'' \leq k$, we have $b'_{i'}<b_{i'}$ and
        $b'_{i''}<b_{i''}$ because colour 1 appears at most one time in
        $S_b$.  Hence, $G$ has two sources, a contradiction. So, we
        have this vertex with colour 1 in $C_a$, and in particular,
        $C_a\neq\emptyset$, so that $\chi(G)\geq k+1$.

        If $C_b\neq \emptyset$ then let $C'$ be the component of
        $G\setminus S_b$ that contains $C_b$.  By
        Claim~\ref{odd-bridcage:c2}, we have $\max(S) < \min(C')$.
        Colours 1 and $3, \dots, (k+1)$ are used in $S_b$ and colour $2$
        is used for $v=\min (C_b)$. This contradicts
        Lemma~\ref{ColClique}.

        If $C_b=\emptyset$, then $b_j=v_n$, but
        $\pi(b_j)=k+1\leq \chi(G)$, contradicting Lemma \ref{Lm2}.
\end{proof}

A graph $H$ is a \textit{flower} in $G$ if:
\begin{itemize}
\item
  $V(H)=C_a\cup C_b\cup C_c\cup C_d\cup_{i=1}^k
  V(P_i)\cup_{i=1}^m V(Q_i)\cup V(P_{ac})\cup V(P_{bd})$ for some
  $k,m\geq 2$.
\item For $i\in \{1,\ldots,k\}$, $P_i$ is a flat path of odd length
  $\geq 3$ with two ends $a_i$ and $b_i$.
\item For $i\in \{1,\ldots,m\}$, $Q_i$ is a flat path of odd length
  $\geq 3$ with two ends $c_i$ and $d_i$.
\item $P_{ac}$ is a flat path of even length $\geq 2$ with two ends
  $a_0$ and $c_0$.
\item $P_{bd}$ is a flat path of even length $\geq 2$ with two ends
  $b_0$ and $d_0$.
\item All $a_i$'s, $b_i$'s, $c_i$'s, $d_i$'s are distinct.
\item $K_a=C_a\cup \{a_0,\ldots,a_k\}$ and
  $K_b=C_b\cup \{b_0,\ldots,b_k\}$ are cliques.
\item $K_c=C_c\cup \{c_0,\ldots,c_m\}$ and
  $K_d=C_d\cup \{d_0,\ldots,d_m\}$ are cliques.
\item If $C_a\neq \emptyset$, $\{a_0,\ldots,a_k\}$ is a clique cutset
  of $G$.
\item If $C_b\neq \emptyset$, $\{b_0,\ldots,b_k\}$ is a clique cutset
  of $G$.
\item If $C_c\neq \emptyset$, $\{c_0,\ldots,c_k\}$ is a clique cutset
  of $G$.
\item If $C_d\neq \emptyset$, $\{d_0,\ldots,d_k\}$ is a clique cutset
  of $G$.
\item These are the only edges in $H$.
\item No vertex in $G\setminus H$ has a neighbor in
  $H\setminus (C_a\cup C_b\cup C_c\cup C_d)$.
\end{itemize}

\begin{lemma} \label{Lm:flower} $G$ does not contain a flower.
\end{lemma}

\begin{proof}
  \setcounter{claim}{0} Suppose $G$ contains a flower $H$ as in the
  definition of a flower. W.l.o.g, let
  $b_j=\max(a_0,\ldots,a_k,b_0,\ldots,b_k,c_0,\ldots,c_m,d_0,\ldots,d_m)$. Let
  $a_i'$, $b_i'$, $c_i'$, $d_i'$ be the unique vertices of degree $2$
  adjacent to $a_i$, $b_i$, $c_i$, $d_i$, respectively. By
  Lemma~\ref{l:endFP}, $b'_j<b_j$. Applying Lemma \ref{Lm6}, there are
  two cases:

\smallskip
  \noindent{\bf Case 1:} $\pi(b_i')=1$ $\forall i\in\{0,\ldots,k\}$.

  We omit the proof in this case, because it is similar to the case
  $\pi(b_i')=1$ in Lemma \ref{Lm:odd-birdcage} (but here, we have to
  consider only two flat paths $P_1$ and $P_2$).

\smallskip
\noindent{\bf Case 2:}
$\pi(b_i')=2$ $\forall i\in\{0,\ldots,k\}$: Then for every
$i\in\{1,\ldots,k\}$, $\pi(a_i')=1$ by the parity of the flat paths.

\begin{claim} \label{flower:c1} If $C_b\neq \emptyset$, then for every
  vertex $v\in C_b$, $v>b_j$.
\end{claim}

We omit the proof, because it is similar to the proof Claim
\ref{odd-bridcage:c2} in the proof of Lemma \ref{Lm:odd-birdcage}.
\epc

\begin{claim} \label{flower:c2}
  $\pi(d_0)=1$.
\end{claim}

We omit the proof, because it is similar to the proof of Claim
\ref{odd-bridcage:c3} in the proof of Lemma \ref{Lm:odd-birdcage}.
\epc

Let $b_l=\min(b_0,\ldots,b_k)$. It is clear that $\pi(b_l)=1$ and
$b_i'<b_i$ for every $i\in\{0,\ldots,k\}\setminus \{l\}$.

\begin{claim} \label{flower:c3} 
  For $i\in\{1,\ldots,m\}$,
  $\pi(d_i')=1$.
\end{claim}

          W.l.o.g suppose that for some $1\leq t\leq m$,
          $\{d_1',\ldots,d_t'\}$ is the subset of vertices of
          $\{d_1',\ldots,d_m'\}$ having colour $2$. Let
          $S=\{d_1,\ldots,d_t\}$. We have $d_i'<d_i$ for
          $i\in \{1,\ldots,t\}$ (otherwise $\pi(d_i')=1$). We also
          have $d_0<d_i$ for $i\in \{1,\ldots,t\}$ since $d_0$ is the
          only vertex of colour $1$ in $K_d$. W.l.o.g, assume that
          $d_t=\max(S)$. There does not exist a vertex $u$ in $K_d$
          such that $u<d_t$ and $\pi(u)=2$ by applying Lemma \ref{Lm6} for $S$. Therefore
          $\min(d_{t+1},\ldots,d_m)>\max(S)$, otherwise one of them
          would receive colour $2$.

          If $b_0=b_l$, applying Lemma \ref{Lm5} for the cutset
          $(S\cup \{b_1,\ldots,b_k\})\cap V(G_{\leq \max(S)})$ of
          $G_{\leq \max(S)}$, there are sources in both sides of this
          cutset, a contradiction.

          If $b_0\neq b_l$, applying Lemma \ref{Lm5} for the cutset
          $(S\cup \{b_0\})\cap V(G_{\leq \max(S)})$ of
          $G_{\leq \max(S)}$, there are sources in both sides of this
          cutset, a contradiction.
   
\epc

	By Claim \ref{flower:c3} and the parity of all the flat paths $Q_i$, $\pi(c_i')=2$ for every $i\in\{1,\ldots,m\}$.

	\begin{claim} \label{flower:c4}
	$\pi(c_0)\geq 2$.
	\end{claim}

	Suppose $\pi(c_0)=1$, then we have $c_0<c_i$ for
        $i\in\{1,\ldots,m\}$ since $c_0$ is the only vertex of colour
        $1$ in $K_c$. We also have $c_i'<c_i$ for $i\in\{1,\ldots,m\}$
        since otherwise some vertex in $\{c_1',\ldots,c_m'\}$ would
        have colour $1$.
	
	If $b_0=b_l$, applying Lemma \ref{Lm5} for the cutset
        $\{c_1,\ldots,c_m,b_1,\ldots,b_k\}$ of $G$, there are sources
        in both sides of this cutset, a contradiction.  
	
	If $b_0\neq b_l$, applying Lemma \ref{Lm5} for the cutset
        $\{c_1,\ldots,c_m,b_0\}$ of $G$, there are sources in both
        sides of this cutset, a contradiction.  \epc

	\begin{claim} \label{flower:c5}
	$\pi(c_0')=1$.
	\end{claim}

	Suppose $\pi(c_0')=2$, then $c_0'<c_0$. There must exist a
        vertex $v$ of colour $1$ in $K_c$ such that $v<c_0$.
	
	If $b_0=b_l$, applying Lemma \ref{Lm5} for the cutset
        $\{c_0,b_1,\ldots,b_k\}$ of $G$, there are sources in both
        sides of this cutset, a contradiction.
	
	If $b_0\neq b_l$, applying Lemma \ref{Lm5} for the cutset
        $\{c_0,b_0\}$ of $G$, there are sources in both sides of this
        cutset, a contradiction.  \epc

        By Claim \ref{flower:c5} and the parity of $P_{ac}$,
        $\pi(a_0')=1$. We must have a vertex of colour $1$ in $C_a$,
        otherwise $a_i'<a_i$ for every $i\in\{1,\ldots,k\}$ and we
        have at least two sources in $G$. Therefore, if
        $C_b\neq \emptyset$ then $b_0,\ldots,b_k$ receive colours $1$
        and $3,\ldots,k+2$ and $\min(C_b)$ receives colour $2$, a
        contradiction to Lemma \ref{ColClique}. If $C_b= \emptyset$
        then $b_j=v_n$, but $\pi(b_j)=k+2\leq \chi(G)$ (we have a
        clique of size $(k+2)$ in $K_a$), contradicting Lemma
        \ref{Lm2}.
\end{proof}

A graph $H$ is a \textit{sun} in $G$ if:
\begin{itemize}
\item $V(H)=V(I)\cup_{i=0}^k K_i$ for some $k\geq 0$.
\item $I$ is a hole.
\item For $i\in\{0,\ldots,k\}$, $a_i$'s, $b_i$'s are distinct vertices of $I$ such that: 
  \begin{itemize}
  \item They appear in the following clock-wise order: $a_0,b_0,\ldots,a_k,b_k$.
  \item For $\forall i\in\{0,\ldots,k\}$, $a_i$ is adjacent to $b_i$.
  \item For $\forall i\in\{0,\ldots,k\}$, the path in $I$ from $b_i$ to $a_{i+1}$ is a flat path of length $\geq 2$ (the subscript is taken modulo $(k+1)$).
  \end{itemize}
\item For $i\in\{0,\ldots,k\}$, $K_i$ is a non-empty clique complete to $\{a_i,b_i\}$.
\item For $i\in\{0,\ldots,k\}$, $\{a_i,b_i\}$ is a cutset of $G$.
\item These are the only edges in $H$.
\item No vertex in $G\setminus H$ has a neighbor in $I$.
\end{itemize}

\begin{lemma} \label{Lm:sun}
$G$ does not contain a sun.
\end{lemma}

\begin{proof}
\setcounter{claim}{0}
Suppose $G$ contains a sun $H$ as in the definition. For $i\in\{0,\ldots,k\}$, let $a_i'$, $b_i'$ be the vertices of degree $2$ adjacent to $a_i$, $b_i$ in $I$. W.l.o.g, suppose $b_k=\max(a_0,b_0,\ldots,a_k,b_k)$. Let $C_1$, $C_2$ be two connected components of $G\setminus \{a_k,b_k\}$ such that $C_2$ contains $K_k$. 

By Lemma \ref{l:endFP}, $b'_k<b_k$. Applying Lemma \ref{Lm7} for the
cutset $\{a_k,b_k\}$, we have $\pi(a_k')=\pi(b_k')=2$ and there exist
some vertex $u\in K_k$, $p,q\in C_2\setminus K_k$ such that $upq$ is a
triangle, $a_k<a_k'$, $u<a_k$ and $u<b_k$. It is clear that
$v_1\in C_2$. If $I$ is an odd hole, then $G$ contains $F_{11}$, a
contradiction. So, the length of $I$ is even. If there exist some
$a_i$, $b_i$, for $i\neq k$ such that the path $b_kIa_i$ and $b_iIa_k$
(taken in clock-wise order) are of odd length, then $G$ contains
$F_{12}$, a contradiction. Hence, the flat path $b_{k-1}Ia_k$ and
$b_kIa_0$ are of even length, and each flat path $b_iIa_{i+1}$ is of
odd length, for $\forall i\in\{0,\ldots,k-2\}$.

\begin{claim} \label{sun:c2}
$\pi(b_{k-1})=\pi(a_{k-1}')=1$.
\end{claim}

Since we have $a_k<a_k'$, the flat path $a_kIb_{k-1}$ (in
counter-clockwise order) is a directed path in $D_G$ from $a_k$ to
$b_{k-1}$. We have that $\min(K_{k-1}\cup a_{k-1})>b_{k-1}$, otherwise
if there exists a vertex $v\in K_{k-1}\cup a_{k-1}$ such that
$v<b_{k-1}$ then $G_{<b_{k-1}}$ is not connected ($b_{k-1}'$ and $v$
are in different connected components). Therefore, $\pi(b_{k-1})=1$
since the flat path $a_kIb_{k-1}$ is of even length. We have
$a_{k-1}<a_{k-1}'$ since otherwise $G_{<a_{k-1}}$ is not connected
($b_{k-1}$ and $a_{k-1}'$ are in different connected
components). Since $\pi(a_{k-1})\geq 2$, $\pi(a_{k-1}')=1$.  \epc

By the same argument as in Claim \ref{sun:c2}, we can prove that:
$\pi(b_{i})=\pi(a_{i}')=1$ for every $i\in\{0,\ldots,k-1\}$. And by
Lemma \ref{C3} and the parity of the flat path $a_0Ib_k$, we have
$\pi(b_k')=1$, contradicting to the fact that $\pi(b_k')=2$ which we
mentioned previously by Lemma \ref{Lm7}.
\end{proof}

\section{Proof of Theorem \ref{T1}} \label{S:4}

For our proof, we need results from~\cite{HW89}.  A graph $G$ is a
\emph{parity graph} if for every pair $u, v \in V(G)$, all induced
paths from $u$ to $v$ have the same parity. A graph is
\emph{distance-hereditary} if for every pair $u, v \in V(G)$, all induced paths from $u$ to $v$ have the same length.  Clearly, every
distance-hereditary graph is a parity graph.  A graph is
\emph{chordal} if it contains no hole. 

\begin{theorem}[see \cite{BM86}]
  \label{th:chgf}
  Every gem-free chordal graph is a distance-hereditary graph and
  therefore a parity graph. 
\end{theorem}

\begin{theorem}[see \cite{HW89}]
  \label{th:HdW}
  Every fish-free parity graph is good. 
\end{theorem}

Throughout the rest of this section, let $G$ be a minimally bad claw-free graph that is not an obstruction.  Our goal is to prove that this implies a contradiction, thus proving Theorem~\ref{T1}.

\begin{lemma} \label{L1}
Let $H=v_0\ldots v_k$ be a hole in $G$ and $u\in V(G)\setminus V(H)$ has some neighbor in $H$. Then $u$ has two or three neighbors in $H$ and they induce a path $($of length $1$ or $2)$, or $u\cup V(H)$ induces a $4$-wheel.
\end{lemma}

\begin{proof}
  If $u$ is adjacent to some vertex $v_i$ in $H$, then $u$ must be
  adjacent also to $v_{i-1}$ or $v_{i+1}$, otherwise
  $\{v_i,v_{i-1},u,v_{i+1}\}$ induces a claw. Suppose $G[N(u)\cap H]$
  induces at least two components, where $\{v_i,\ldots,v_j\}$ and
  $\{v_k,\ldots,v_m\}$ are its two consecutive components in $H$ (in
  clock-wise order, $i,j,k,m$ are distinct) then
  $\{u,v_{j-1},v_j,\ldots,v_k,v_{k+1}\}$ induces $F_2$, a
  contradiction. Then $G[N(u)\cap H]$ induces only one component.
\begin{itemize}
	\item If $|H|\geq 5$: If $u$ has at least four neighbors on
          $H$ then $u$ and its four consecutive neighbors in $H$
          induce a gem (a special case of $F_2$), a contradiction. Then $u$ has two or three neighbors in $H$ and they induce a path.
	\item If $|H|=4$: $u$ can have two or three neighbors in $H$ and they induce a path or $u$ is complete to $H$ and $u\cup V(H)$ induces a $4$-wheel.
\end{itemize}
\end{proof}

\begin{lemma} \label{L2}
Let $H=v_0\ldots v_k$ be a hole in $G$. For $i\in\{0,\ldots,k\}$, let $S_i=\{u| N(u)\cap H=\{v_{i-1},v_i,v_{i+1}\}\}$ and $R_i=\{u|N(u)\cap H=\{v_i,v_{i+1}\}\}$. Then for any $i\in\{0,\ldots,k\}$:
\begin{enumerate}
	\item $S_i$ is a clique.
	\item $R_i$ is a clique.
	\item $S_i$ is complete to $S_{i+1}$ and anticomplete to $S_j$ for any $j\notin\{i-1,i,i+1\}$.
	\item $R_i$ is anticomplete to $R_j$ for any $j\neq i$.
	\item $S_i$ is complete to $R_{i-1}$ and $R_i$.
	\item If $S_i\neq \emptyset$ then $R_j=\emptyset$ for any $j\notin\{i-1,i\}$.
	\item If $R_i\neq \emptyset$ then $R_j=\emptyset$ for any $j\in\{i-2,i-1,i+1,i+2\}$.
\end{enumerate}
\end{lemma}

\begin{proof}
We prove each statement following the index.
\begin{enumerate}
\item Let $a,b\in S_i$. If $a$ is not adjacent to $b$ then
  $\{v_{i+1},a,b,v_{i+2}\}$ induces a claw, a contradiction.
\item Let $a,b\in R_i$. If $a$ is not adjacent to $b$ then
  $\{v_{i+1},a,b,v_{i+2}\}$ induces a claw, a contradiction.
	\item Let $a\in S_i$, $b\in S_{i+1}$ and $c\in S_j$. If $a$ is
          not adjacent to $b$ then $\{v_i,v_{i-1},a,v_{i+1},b\}$
          induces a gem, a contradiction. If $a$ is adjacent to $c$ then
          $\{a,c,v_i,v_{i+1},\ldots,v_{j-1},v_j\}$ induces $F_3$ or a
          gem, a contradiction.
	\item Let $a\in R_i$ and $b\in R_j$. If $a$ is adjacent to $b$
          then $\{a,b,v_i,v_{i+1},\ldots,v_j,v_{j+1}\}$ contains $F_3$
          or a gem, a contradiction.
	\item Let $a\in S_i$ and $b\in R_i$. If $a$ is not adjacent to
          $b$ then $\{v_i,v_{i-1},a,v_{i+1},b\}$ induces a gem, a contradiction. $S_i$ is also complete to $R_{i-1}$ by symmetry.
	\item Let $a\in S_i$ and $b\in R_j$. If $a$ is adjacent to $b$
          then $\{a,b,v_i,v_{i+1},\ldots,v_j,v_{j+1}\}$ contains $F_3$
          or a gem, a contradiction. If $a$ is not adjacent to $b$ then
          $\{a,b\}\cup V(H)$ contains $F_1$, $F_5$, $F_6$ or a fish (a
          special case of $F_{11}$).  So,
          $R_j=\emptyset$.
	\item Let $a\in R_i$, $b\in R_{i+1}$ and $c\in R_{i+2}$. If
          $a$ is not adjacent to $b$, then $\{a,b\}\cup V(H)$ induces
          $F_2$, a contradiction. If $a$ is adjacent to $b$, then
          $\{v_{i+1},v_i,a,b,v_{i+2}\}$ induces a gem,
          a contradiction. So,
          $R_{i+1}=\emptyset$. If $a$ is not adjacent to $c$, then
          $\{a,c\}\cup V(H)$ induces $F_3$. If $a$ is adjacent to $c$,
          then $\{a,c,v_i,v_{i+1},v_{i+2},v_{i+3}\}$ induces either
          $F_3$ or $F_4$, a contradiction. So,
          $R_{i+2}=\emptyset$. The proof for $R_{i-2}$ and $R_{i-1}$
          is similar.
\end{enumerate}
\end{proof}

\begin{lemma} \label{L3}
$G$ does not contain a $4$-wheel.
\end{lemma}

\begin{proof}
  \setcounter{claim}{0} Suppose $G$ contains a $4$-wheel consisting of
  a hole $H=v_0v_1v_2v_3$ and a vertex $x$ complete to that hole. For
  $i\in\{0,\ldots,3\}$, let
  $S_i=\{u| N(u)\cap H=\{v_{i-1},v_i,v_{i+1}\}\}$,
  $R_i=\{u|N(u)\cap H=\{v_i,v_{i+1}\}\}$ and
  $T=\{u|N(u)\cap H=\{v_0,v_1,v_2,v_3\}\}$. Note that $x\in T$.
\begin{claim}
For $i\in\{0,\ldots,3\}$, $R_i=\emptyset$.
\end{claim} 
Let $a\in R_i$. If $a$ is not adjacent to $x$, then
$\{v_i,v_{i-1},x,v_{i+1},a\}$ induces a gem, a contradiction. If $a$
is adjacent to $x$, then $\{x,v_{i-2},v_{i-1},v_i,a\}$ induces a gem,
a contradiction. So, $R_i=\emptyset$.  \epc

\begin{claim}
For $i\in\{0,\ldots,3\}$, $S_i$ is complete to $T$.  
\end{claim}

Let $a\in S_i$ and $b\in T$. If $a$ is not adjacent to $b$ then
$\{v_{i-1},v_{i-2},b,v_i,a\}$ induces a gem, a contradiction.
\epc

By Lemma \ref{L2}, $S_i$ is a clique complete to $S_{i+1}$ and
anticomplete to $S_{i+2}$. So, $G[\cup_{i=0}^3 S_i]$ is $P_4$-free.

\begin{claim}
$V(G)=T\cup_{i=0}^3 S_i\cup V(H)$.
\end{claim}

Suppose there exists some vertex $v\in V(G)\setminus (T\cup_{i=0}^3
S_i\cup V(H))$ and $v$ has some neighbors in $T\cup_{i=0}^3 S_i\cup
V(H)$. If $v$ has a neighbor in $V(H)$, then $v\in T\cup_{i=0}^3 S_i$,
a contradiction. If $v$ has a neighbor $u$ in $S_i$ for some $i$ but
no neighbor in $V(H)$, then $\{u,v,v_{i-1},v_{i+1}\}$ forms a claw, a
contradiction. If $v$ has a neighbor $u$ in $V(T)$ but no neighbor in
$V(H)$, then $\{u,v_0,v_2,v\}$ forms a claw, a contradiction.
\epc

We have $G[T]$ is $P_4$-free (otherwise $v_0$ and some $P_4$ in $G[T]$
form a gem), $G[\cup_{i=0}^3 S_i\cup V(H)]$ is $P_4$-free and $T$ is
complete to $\cup_{i=0}^3 S_i\cup V(H)$, then $G$ is $P_4$-free and
therefore is good by Theorem~\ref{th:Ccograph}, a contradiction.
\end{proof}

A \emph{boat} is a graph consisting of a hole $H$ and a vertex $x$ has three consecutive neighbors on $H$.

\begin{lemma} \label{L4}
$G$ does not contain a boat.
\end{lemma}

\begin{proof}
  \setcounter{claim}{0} Suppose $G$ contains a boat consisting of a
  hole $H=v_0\ldots v_k$ and a vertex $x$ that has three neighbors on
  $H$: $v_0$, $v_1$, $v_2$. We can assume that $H$ is an even hole
  since otherwise $G$ contains $F_1$. For $i\in\{0,\ldots,k\}$, let
  $S_i=\{u| N(u)\cap H=\{v_{i-1},v_i,v_{i+1}\}\}$ and
  $R_i=\{u|N(u)\cap H=\{v_i,v_{i+1}\}\}$. We consider two cases:

\smallskip
\noindent{\bf Case 1:} There exists some vertex $y\in R_j$ for some
$j$. Then $j\in\{0,1\}$ by Lemma \ref{L2}. W.l.o.g, suppose $j=1$ then
$R_0=\emptyset$ by Lemma \ref{L2}. Also by Lemma \ref{L2},
$S_i=\emptyset$ for $i\notin\{1,2\}$. Now,
$S_1\cup S_2\cup R_1\cup V(H)$ forms a cap in $G$ with
$L=S_1\cup\{v_1\}$, $R=S_2\cup\{v_2\}$, $C=R_1$ and $P$ is the flat
path from $v_0$ to $v_3$ in $H$, contradicting Lemma \ref{Lm:cap}.
	
\smallskip {\noindent \bf Case 2:} For every $i\in\{0,\ldots,k\}$,
$R_i=\emptyset$. By Lemma \ref{L2}, each $S_i$ is a clique complete to
$S_{i+1}$. It is clear that $V(G)=\cup_{i=0}^k S_i\cup V(H)$. Hence, $G$ is a
parity graph. Also, $G$ is fish-free because the fish is an
obstruction. Therefore, $G$ is good by Theorem~\ref{th:HdW}, a
contradiction.
\end{proof}

From now on, by Lemmas \ref{L3} and \ref{L4}, the neighborhood of any
vertex on a hole in $G$ induces an edge.

\begin{lemma} \label{L5}
$G$ does not contain a short prism.
\end{lemma}

\begin{proof}
Otherwise, $G$ contains $F_3$ or $F_4$, a contradiction.
\end{proof}

From now on, we know that that every prism in $G$ is not short, or in
another word, its three paths are of length at least two.

\begin{lemma} \label{L6} $G$ does not contain an imparity prism.
\end{lemma}

\begin{proof}
Otherwise, it contains $F_7$, a contradiction.
\end{proof}

A graph $H$ is a \textit{prism system} if:
\begin{itemize}
	\item $V(H)=\cup_{i=1}^k V(P_i)$ for some $k\geq 3$.
	\item $\forall i\in\{1,\ldots,k\}$, $P_i$ is a path of length $\geq 2$ with two ends $a_i$, $b_i$. All $P_i$'s are disjoint.
	\item $S_a=\{a_1,\ldots,a_k\}$ and $S_b=\{b_1,\ldots,b_k\}$ are two cliques.
	\item These are the only edges in $H$.
\end{itemize}
Note that if $k=3$ then $H$ is simply a prism. A prism system is \textit{even} (\textit{odd}) if the lengths of all path $P_i$'s are even (odd). 

\begin{lemma} \label{L7}
$G$ does not contain an even prism.
\end{lemma}

\begin{proof}
\setcounter{claim}{0}
Suppose $G$ contains an even prism. Then there exists an even prism system in $G$ as in the description, choose such a prism system $H$ with maximum value of $k$. Let $C_a=\{v\in V(G)\setminus V(H)|N(v)\cap S_a\neq \emptyset\}$ and $C_b=\{v\in V(G)\setminus V(H)|N(v)\cap S_b\neq \emptyset\}$.    

\begin{claim} \label{L7:c1}
All paths $P_i$'s are flat.
\end{claim}

If there exists a vertex $v\in V(G)\setminus V(H)$ which has some
neighbors $\{a,b\}$ in the internal of some path $P_i$, then $G$
contains $F_2$, $F_3$ or $F_8$, a contradiction. 
\epc

\begin{claim} \label{L7:c2}
$C_a$ is a clique complete to $S_a$ and $C_b$ is a clique complete to $S_b$.
\end{claim}

Follows directly from Lemma \ref{L2}.
\epc

\begin{claim} \label{L7:c3}
We have the followings:
\begin{enumerate}
	\item If $C_a\neq \emptyset$, $S_a$ is a clique cutset of $G$.
	\item If $C_b\neq \emptyset$, $S_b$ is a clique cutset of $G$.
\end{enumerate}
\end{claim}

We prove only the first statement, the second is similar. Suppose that
there exists a path $P$ from some vertex in $C_a$ to some
vertex in $C_b$. The length of $P$ is even otherwise $G$ contains
$F_7$, then $H\cup P$ is a bigger even prism system, a contradiction to the choice of $H$. 
\epc

By Claims \ref{L7:c1}, \ref{L7:c2} and \ref{L7:c3}, $H\cup C_a\cup C_b$ forms an even birdcage, contradicting Lemma~\ref{Lm:even-birdcage}. 

\end{proof}

A \emph{bracelet}  (see Figure \ref{F:2}) has $6$ paths of
length $\geq 2$: two paths in the sides are of even length; the other
four paths are of odd length. 

A graph $H$ is a \emph{bracelet system} if:
\begin{itemize}
	\item $V(H)=\cup_{i=1}^k V(P_i)\cup_{i=1}^m V(Q_i)\cup V(P_{ac})\cup V(P_{bd})$ for some $k,m\geq 2$.
	\item For $i\in \{1,\ldots,k\}$, $P_i$ is a path of odd length $\geq 3$ with two ends $a_i$ and $b_i$.
	\item For $i\in \{1,\ldots,m\}$, $Q_i$ is a path of odd length $\geq 3$ with two ends $c_i$ and $d_i$.
	\item $P_{ac}$ is a path of even length $\geq 2$ with two ends $a_0$ and $c_0$.
	\item $P_{bd}$ is a path of even length $\geq 2$ with two ends $b_0$ and $d_0$.
	\item All path $P_i$'s, $Q_i$'s, $P_{ac}$, $P_{bd}$ are disjoint.
	\item $S_a=\{a_0,\ldots,a_k\}$ and $S_b=\{b_0,\ldots,b_k\}$ are cliques.
	\item $S_c=\{c_0,\ldots,c_m\}$ and $S_d=\{d_0,\ldots,d_m\}$ are cliques.
	\item These are the only edges in $H$.
\end{itemize}
Note that if $k=m=2$, then $H$ is simply a bracelet.

\begin{lemma} \label{L8}
$G$ does not contain an odd prism.
\end{lemma}

\begin{proof}
\setcounter{claim}{0}
  Suppose $G$ contains an odd prism. We consider the following cases:

\smallskip
  {\noindent \bf Case 1:} $G$ contains a bracelet.
  Then there exists a bracelet system in $G$ as in the description,
  choose such a system $H$ with maximum value of $k+m$. Let
  $C_a=\{v\in V(G)\setminus V(H)|N(v)\cap S_a\neq \emptyset\}$,
  $C_b=\{v\in V(G)\setminus V(H)|N(v)\cap S_b\neq \emptyset\}$,
  $C_c=\{v\in V(G)\setminus V(H)|N(v)\cap S_c\neq \emptyset\}$ and
  $C_d=\{v\in V(G)\setminus V(H)|N(v)\cap S_d\neq \emptyset\}$.

  \begin{claim} \label{L8:c1} All the paths $P_i$'s, $Q_i$'s,
    $P_{ac}$, $P_{bd}$ are flat.
  \end{claim}

  Suppose there is some vertex $v\in V(G)\setminus V(H)$ has some
  neighbor in the internal of one of these paths. If $v$ has some
  neighbor on $P_{ac}$ or $P_{bd}$, then $G$ contains $F_9$,
  a contradiction. If $v$ has some neighbor on some $P_i$ or $Q_i$, then
  $G$ contains $F_2$, $F_3$, $F_9$ or $F_{10}$, a contradiction.  \epc

  \begin{claim} \label{L8:c2} $C_a$ is a clique complete to $S_a$.
  \end{claim}

  Follows directly from Lemma \ref{L2}.  \epc

  We have a similar statement for $C_b$, $C_c$ and $C_d$.
  
\begin{claim} \label{L8:c3} If $C_a\neq \emptyset$, $S_a$ is a
    clique cutset in $G$.
  \end{claim}

  Otherwise, if there exists some  path $P$ from a vertex in
  $C_a$ to some vertex in $C_c$ or $C_d$, then $G$ contains an even
  prism, contradicting Lemma \ref{L7}. If there exists some induced
  path $P$ from a vertex in $C_a$ to some vertex in $C_b$, then $P$ is
  of odd length and therefore $H\cup P$ is a bigger bracelet system,
  a contradiction to the choice of $H$.  \epc

  We also have similar statement for $S_b$, $S_c$ and $S_d$. By Claims
  \ref{L8:c1}, \ref{L8:c2} and \ref{L8:c3},
  $H\cup C_a\cup C_b\cup C_c\cup C_d$ is a flower in $G$,
  contradicting Lemma \ref{Lm:flower}.

\smallskip
  {\noindent\bf Case 2:} $G$ does not contain bracelet. There exists
  an odd prism system in $G$ as in the description, choose such a
  prism system $H$ with maximum value of $k$. Let
  $C_a=\{v\in V(G)\setminus V(H)|N(v)\cap S_a\neq \emptyset\}$ and
  $C_b=\{v\in V(G)\setminus V(H)|N(v)\cap S_b\neq \emptyset\}$.
  \begin{claim} \label{L8:c4} Let $v\in V(G)\setminus V(H)$ be a
    vertex has some neighbor $\{a,b\}$ in the internal of some path
    $P_i$ ($a$ is closer to $a_i$ than $b$ in $P_i$). Then two paths
    $a_iP_ia$ and $b_iP_ib$ are of even length $\geq 2$.
  \end{claim}

  Otherwise $G$ contains $F_2$, $F_3$ or $F_9$, a contradiction.  \epc

  \begin{claim} \label{L8:c5} If $v\in V(G)\setminus V(H)$ has some
    neighbors in $P_i$, then all path $P_j$'s are flat for any
    $j\neq i$.
  \end{claim}

  Otherwise $G$ contains $F_{10}$, a contradiction.  \epc W.l.o.g,
  suppose that $P_1$ is the only path among $P_i$'s which might not be
  flat. For some $m\geq0$, let $\{c_1,d_1\},\ldots,\{c_m,d_m\}$ be all
  the possible positions in $P_1$ to which a vertex
  $v\in V(G)\setminus V(H)$ can be adjacent ($c_id_i$ is an edge; all
  vertices are listed in order from $a_1$ to $b_1$). For
  $i\in\{1,\ldots,m\}$, let
  $K_i=\{v\in V(G)\setminus V(H)|N(v)\cap \{c_i,d_i\}\neq
  \emptyset\}$.

  \begin{claim} \label{L8:c6} $C_a$ is a clique complete to $S_a$;
    $C_b$ is a clique complete to $S_b$ and $K_i$ is a clique complete
    to $\{c_i,d_i\}$ for $i\in\{1,\ldots,m\}$.
  \end{claim}

  Follows directly from Lemma \ref{L2}.  \epc

  \begin{claim} \label{L8:c7} $a_1P_1c_1$ and $d_mP_1b_1$ are flat
    paths of even length $\geq 2$; $d_iP_1c_{i+1}$ is a flat path of
    odd length $\geq 3$ for $i\in\{1,\ldots,m-1\}$.
  \end{claim}

  Follows from Claim \ref{L8:c4}.  \epc

  \begin{claim} \label{L8:c8} If $C_a$ and $C_b\neq \emptyset$, $S_a$,
    $S_b$, $\{c_i,d_i\}$ are clique cutsets in $G$ for
    $i\in\{1,\ldots,m\}$.
  \end{claim}

  Otherwise, if there is a path from some vertex in $K_i$ to
  some vertex in $K_j$ for some $i\neq j$, then $G$ contains a
  bracelet, a contradiction. If there is a path from some
  vertex in $C_a$ to some vertex in $K_i$ for some $i$, then $G$
  contains $F_7$ or an even prim, a contradiction. If there is a path $P$ from a vertex in $C_a$ to some vertex in $C_b$,
  then $H\cup P$ is a bigger odd prism system, a contradiction to the
  choice of $H$.  \epc

  By Claims \ref{L8:c5}, \ref{L8:c6}, \ref{L8:c7} and \ref{L8:c8},
  $H\cup C_a\cup C_b\cup_{i=1}^m K_i$ forms an odd birdcage in $G$,
  contradicting Lemma \ref{Lm:odd-birdcage}.
\end{proof}

By Lemmas \ref{L5}, \ref{L6}, \ref{L7} and \ref{L8}, $G$ is
prism-free.

\begin{lemma} \label{L9} 
  $G$ does not contain a hole.
\end{lemma}

\begin{proof}
  \setcounter{claim}{0}
  Suppose $G$ contains a hole $I$. For some $k\geq 0$, let
  $\{a_0,b_0\},\ldots,\{a_k,b_k\}$ be all the possible positions in
  $I$ to which a vertex $v\in V(G)\setminus I$ can be adjacent
  ($a_ib_i$ is an edge; all the vertices are listed in clock-wise
  order). For $i\in\{0,\ldots,k\}$, let
  $K_i=\{v\in V(G)\setminus I|N(v)\cap \{a_i,b_i\}\neq \emptyset\}$.

  \begin{claim} \label{L9:c1}
    For $i\in\{0,\ldots,k\}$, $K_i$ is a clique complete to $\{a_i,b_i\}$.
  \end{claim} 

  Follows directly from Lemma \ref{L2}.
\epc

\begin{claim} \label{L9:c2}
  For $i\in\{0,\ldots,k\}$, $\{a_i,b_i\}$ is a cutset of $G$.
\end{claim}

Otherwise, there is a path from a vertex in $K_i$ to some
vertex in $K_j$, for some $j\neq i$, so $G$ contains a prism,
a contradiction.  \epc

By Claims \ref{L9:c1} and \ref{L9:c2}, $I\cup_{i=0}^k K_i$ forms a sun
in $G$, contradicting Lemma~\ref{Lm:sun}.
\end{proof}

By Lemmas \ref{L9}, $G$ is chordal. And since $G$ is gem-free, $G$ is
a parity graph by Theorem~\ref{th:chgf}. Hence $G$ is a fish-free
parity graph and is therefore good by Theorem~\ref{th:HdW}, a
contradiction.  This proves that every minimally bad claw-free graph is an obstruction. 

To prove Theorem \ref{T1}, we are left to prove that every obstruction is a minimally bad claw-free graph. Suppose that it is not true for some graph $F$ in the list of obstructions. Since $F$ is bad (as we already specify a bad order for every obstruction), $F$ must contain a minimally bad claw-free graph $F'$ as an induced subgraph. Since every minimally bad claw-free graph is an obstruction, $F'$ is also an obstruction. However, it is easy to check that there do not exist two obstructions in our list such that one contains the other as an induced subgraph, contradiction. 

\section{Conclusion}

In this paper, we give the characterization of good claw-free graphs in terms of minimal forbidden induced subgraphs. Note that the arguments in Sections \ref{S:3} and \ref{S:4} can be turned into a polynomial algorithm for recognizing this class, where each structure in Section \ref{S:3} corresponds to a kind of decomposition. A full characterization of good graphs seems hard to achieve, as we observe that the actual structure of minimally bad graphs could be much more complicated. The following question is open: 
 
\noindent \textbf{Open question.} Is the chromatic number of every minimally bad graph $3$?  

We see that this is true for claw-free graphs. The next step would be finding the characterization for good perfect graphs, or some interesting subclasses of perfect graphs.

\bibliographystyle{abbrv}
\bibliography{grundy}

\begin{thebibliography}{1}

\bibitem{BT94}
L.~Babel and G.~Tinhofer.
\newblock Hard-to-color graphs for connected sequential colorings.
\newblock {\em Discrete Applied Mathematics}, 51(1-2):3--25, 1994.

\bibitem{BM86}
H.~Bandelt and H.~M. Mulder.
\newblock Distance-hereditary graphs.
\newblock {\em J. Comb. Theory, Ser. {B}}, 41(2):182--208, 1986.

\bibitem{BCD14}
F.~Benevides, V.~Campos, M.~Dourado, S.~Griffiths, R.~Morris, L.~Sampaio, and
  A.~Silva.
\newblock Connected greedy colourings.
\newblock In {\em Latin American Symposium on Theoretical Informatics}, pages
  433--441. Springer, 2014.

\bibitem{BFKS15}
{\'E}.~Bonnet, F.~Foucaud, E.~J. Kim, and F.~Sikora.
\newblock Complexity of grundy coloring and its variants.
\newblock In {\em International Computing and Combinatorics Conference}, pages
  109--120. Springer, 2015.

\bibitem{CS79}
C.~A. Christen and S.~M. Selkow.
\newblock Some perfect coloring properties of graphs.
\newblock {\em Journal of Combinatorial Theory, Series B}, 27(1):49--59, 1979.

\bibitem{W90}
D.~de~Werra.
\newblock Heuristics for graph colorings.
\newblock In {\em Computational Graph Theory, Computing Supplementum},
  volume~7, pages 191--208. Springer, 1990.

\bibitem{HW89}
A.~Hertz and D.~de~Werra.
\newblock Connected sequential colorings.
\newblock {\em Discrete mathematics}, 74(1-2):51--59, 1989.

\bibitem{S74}
D.~Seinsche.
\newblock On a property of the class of $n$-colorable graphs.
\newblock {\em Journal of Combinatorial Theory, Series B}, 16:191--196, 1974.

\end{thebibliography}

\end{document}